\documentclass[11pt]{article}
\usepackage{times}
\usepackage{helvet}
\usepackage{courier}
\usepackage{graphicx}
\usepackage{wrapfig}
\usepackage{comment}

\usepackage{amsmath}
\usepackage{amsthm}
\usepackage{amssymb}

\usepackage{todonotes}

\newcommand{\pref}{\succ}
\newcommand{\score}{{{\mathrm{score}}}}
\newcommand{\naturals}{{{\mathbb{N}}}}

\newcommand{\np}{{\mathrm{NP}}}
\newcommand{\fpt}{{\mathrm{FPT}}}

\newtheorem{theorem}{Theorem}

\newtheorem{corollary}[theorem]{Corollary}

\newtheorem{remark}{Remark}
\newtheorem{definition}{Definition}
\newtheorem{example}{Example}
\newtheorem{proposition}[theorem]{Proposition}

\newcommand{\calR}{\mathcal{R}}

\newcommand{\cc}{{{{\mathrm{CC}}}}}

\newcommand{\sntv}{{{{\mathrm{SNTV}}}}}
\newcommand{\perf}{{{{\mathrm{Perf}}}}}
\newcommand{\bloc}{{{{\mathrm{Bloc}}}}}
\newcommand{\kborda}{{{{k\hbox{-}\mathrm{Borda}}}}}
\DeclareMathOperator*{\argmin}{arg\,min}

\newcommand{\pos}{{{{\mathrm{pos}}}}}

\def\N{\mathbb{N}}
\def\row#1#2{{#1}_1,\ldots ,{#1}_{#2}}

\newcommand{\fourvote}[4]{#1 \pref #2 \pref #3 \pref #4}
\newcommand{\shortcite}[1]{\cite{#1}}

\title{Multiwinner Analogues of Plurality Rule: \\ Axiomatic and Algorithmic Perspectives}

\author{
  \makebox[0.35\linewidth]{Piotr Faliszewski} \\ 
  AGH University \\ 
  Krakow, Poland 
\and 
  \makebox[0.35\linewidth]{Piotr Skowron} \\
  University of Oxford \\
  Oxford, UK
\and 
  \makebox[0.35\linewidth]{Arkadii Slinko} \\
  University of Auckland \\
  Auckland, New Zealand
\and 
  \makebox[0.35\linewidth]{Nimrod Talmon\footnote{Most of the work was done while the author was affiliated with TU Berlin (Berlin, Germany).}} \\
  Weizmann Institute of Science \\
  Rehovot, Israel
}

\begin{document}

\maketitle

\begin{abstract}
  We characterize the class of committee scoring rules that satisfy
  the fixed-majority criterion. In some sense, the committee scoring
  rules in this class are multiwinner analogues of the single-winner
  Plurality rule, which is uniquely characterized as the only
  single-winner scoring rule that satisfies the simple majority
  criterion. We define top-$k$-counting committee scoring rules and
  show that the fixed majority consistent rules are a subclass of the
  top-$k$-counting rules. We give necessary and sufficient conditions
  for a top-$k$-counting rule to satisfy the fixed-majority criterion.
  We find that, for most of the rules in our new class, the complexity
  of winner determination is high (that is, the problem of computing
  the winners is $\np$-hard), but we also show examples of rules with
  polynomial-time winner determination procedures. For some of the
  computationally hard rules, we provide either exact FPT algorithms
  or approximate polynomial-time algorithms.

%
\end{abstract}

\section{Introduction}
The scoring rules in general, and Plurality specifically,
are among the most often used and best studied single-winner voting rules.
Recently, multiwinner analogues of scoring rules---called committee
scoring rules---were introduced by Elkind et
al.~\cite{elk-fal-sko-sli:c:multiwinner-properties}, but our
understanding of them is so far quite limited.  In this paper, we seek
to somewhat rectify this situation by asking a seemingly
innocuous question: \emph{Among the committee scoring rules, which one
  is the analogue of Plurality?}  (The single-winner
Plurality rule elects the candidate who is listed as the most
preferred one by the largest number of voters.)  Using an axiomatic
approach, we find a rather surprising answer. Not only is there a
whole class of committee scoring rules that can be viewed as
corresponding to Plurality, but also one of the most
`obvious' candidates to be `the multiwinner Plurality'---the single
non-transferable vote rule (or SNTV, see the descriptions later)---falls
short of satisfying our criterion.  On the other hand, it turns out
that the Bloc rule is quite satisfying as `the multiwinner Plurality.'
Yet, 
it certainly is
not the only rule from our class and we believe that the other ones deserve
attention as well.

We complement our axiomatic study with an algorithmic analysis of this
new class of committee scoring rules. In particular, we show that it
can be seen as a subfamily of the OWA-based rules\footnote{OWA stands
  for ordered weighted average.} of Skowron, Faliszewski, and
Lang~\shortcite{sko-fal-lan:c:collective} (also studied by Aziz et
al.~\shortcite{azi-gas-gud-mac-mat-wal:c:approval-multiwinner,azi-bri-con-elk-fre-wal:c:justified-representation};
see also the work of
Kilgour~\shortcite{kil:chapter:approval-multiwinner} for a more
general overview of approval-based multiwinner rules and the work of
Elkind and Ismaili~\shortcite{conf/aldt/ElkindI15} for a different
OWA-based approach to multiwinner voting). However, the hardness
results for general OWA-based rules do not translate directly to our
case (and, indeed, some do not even hold).  
\smallskip

Let us now describe our framework more precisely. In the multiwinner
elections that we consider, each voter ranks the candidates from the
most desired one to the least desired one, and a voting rule is used
to pick a committee of a given size $k$ that, in some sense, best
reflects the voters' preferences.  Naturally, the exact meaning of the
phrase `best reflects' depends strongly on the application at hand, as
well as on the societal conventions and understanding of fairness. For
example, if we are to choose a size-$k$ parliament, then it is
important to guarantee proportional representation of certain
categories of the electorate (for example, political parties, ethnic or
gender groups, etc.); if the goal is to pick a group of products to
offer to customers, then it might be important to maintain
diversity of the offer; if we are to shortlist a group of candidates
for a job, then it is important to focus on
the quality of the selected candidates regardless of how similar some
of them might be 
(see, for example, the discussions provided by Lu and
Boutilier~\cite{bou-lu:c:chamberlin-courant,bou-lu:c:value-directed-cc},
Elkind et al.~\cite{elk-fal-sko-sli:c:multiwinner-properties}, and
Skowron et al.~\cite{sko-fal-lan:c:collective}).
%

In effect, there is quite a variety of multiwinner voting rules. For
example, under SNTV, the winning committee consists of $k$
candidates who are ranked first more frequently than all the others; under the
Bloc rule, each voter gives one point to each candidate he or she
ranks among his or her top $k$ positions, and the committee consists
of $k$ candidates with the most points; under the Chamberlin--Courant
rule, the winning committee consists of $k$ candidates such that each
voter ranks his or her most preferred committee member as highly as
possible (for the exact definition, we point the reader to
Section~\ref{sec:prelim}, to the original paper of Chamberlin and
Courant~\shortcite{cha-cou:j:cc}, or to papers studying the properties and the computational complexity
of this
rule~\cite{pro-ros-zoh:j:proportional-representation,bou-lu:c:chamberlin-courant,elk-fal-sko-sli:c:multiwinner-properties,conf/aldt/ElkindI15,sko-fal-sli:j:multiwinner,sko-fal:c:maxcover}).

It turns out that the three rules mentioned above are examples of
committee scoring rules, a class of rules which generalize
single-winner scoring rules to the multiwinner setting, recently
introduced by Elkind et
al.~\shortcite{elk-fal-sko-sli:c:multiwinner-properties} (see
Section~\ref{sec:prelim} for the definition).\footnote{Naturally,
  these rules were known much before Elkind et
  al.~\shortcite{elk-fal-sko-sli:c:multiwinner-properties} introduced
  the unifying framework for them.} Of course, there are natural
multiwinner rules that cannot be expressed as committee scoring rules,
such as the single transferable vote rule (STV), the Monroe
rule~\cite{mon:j:monroe}, or all the multiwinner rules based on the
Condorcet principle (see, for example, the works of Elkind et
al.~\shortcite{elk-lan-saf:c:condorcet-sets},
Fishburn~\shortcite{fis:j:condorcet-committee}, and
Gehrlein~\shortcite{geh:j:condorcet-committee}). Nonetheless, we
believe that committee scoring rules form a very diverse class of
voting rules that deserves a further study.

In this paper, we analyze committee scoring rules in our search of
a `multiwinner analogue' of the (single-winner) Plurality rule.
Intuitively, it might seem that SNTV, which picks $k$
candidates with best Plurality scores, is such a rule, thus rendering
this question trivial.  However, instead of following this intuition
we take an axiomatic approach. We note that Plurality is the only
single-winner scoring rule that satisfies the simple majority\footnote{In the literature, simple majority is often referred to as majority. However, we write `simple majority' to clearly distinguish it from qualified majority and from fixed-majority.} criterion,
which stipulates that a candidate ranked first by more than half of all 
voters must be a unique winner of the election. We ask for a
committee scoring rule that satisfies the fixed-majority criterion, a
multiwinner analogue of the simple majority criterion introduced by
Debord~\shortcite{deb:j:prudent}. It requires that if
there is a simple majority of the voters, each of whom ranks the same $k$
candidates in their top $k$ positions (perhaps in a different order),
then these $k$ candidates should form a unique winning committee.

With a moment of thought, one can verify that the Bloc rule satisfies
the fixed-majority criterion. It turns out, however, that Bloc is by
far not the only committee scoring rule having this property, and that there
is a whole class of them. We provide an (almost) full characterization
of this class\footnote{For technical reasons, we consider the case
  where there are at least twice as many candidates as the size of the
  committee (that is, $m \geq 2k$).} and analyze the computational complexity of winner
determination for the rules in this class.  Initially, we identify a
somewhat larger class (in terms of strict containment) of \emph{top-$k$-counting} rules for which the
score that a committee receives from a given voter is a function of
the number of committee members that this voter ranks in the top $k$
positions (recall that $k$ is the committee size); we refer to this
function as the \emph{counting function}.  We obtain the following
main results:
\begin{enumerate}
\item We prove that all committee scoring rules that satisfy the fixed-majority criterion are top-$k$-counting rules. 
  We establish conditions 
  on the counting function that are
  necessary and sufficient for the corresponding top-$k$-counting rule
  to satisfy the fixed-majority criterion.  These conditions are a
  fairly mild relaxation of convexity.  In particular, if the counting
  function is convex, then the corresponding top-$k$-counting rule
  satisfies the fixed-majority criterion.  


\item For a large class of counting functions, top-$k$-counting rules
  are $\np$-hard to compute (for example, we show an example of a rule
  that closely resembles the Bloc rule and is hard even to
  approximate). There are, however, some polynomial-time computable
  ones (for example, the Bloc and Perfectionist rules; where the
  latter one is introduced in this paper).

\item If the counting function is concave, then the rule it defines
  fails the fixed-majority criterion, but the rule seems to be easier
  computationally than in the convex case.  For this case, we show a
  polynomial-time ${(1-\frac{1}{e})}$-approximation algorithm as well
  as show fixed-parameter tractability with respect to the number of
  voters for the problem of finding the winning committee.

  
\end{enumerate}
All in all, we find that there is no single multiwinner analogue of Plurality, even if we restrict ourselves to polynomial-time
computable committee scoring rules. Indeed, on the intuitive level
SNTV is such a rule, and through our axiomatic consideration we show
that Bloc and Perfectionist are also good candidates.

Before we move on to the technical content of the paper, let us take a
step back and answer one principal question: \emph{why did we look for the
multiwinner analogue of Plurality?}
We believe that this quest deserves attention for the following two main reasons.
\begin{enumerate}

\item Until recently, multiwinner voting received only passing
  attention in the computational social choice literature, and it is
  still receiving only moderate attention in the social choice
  literature.  In effect, our understanding of multiwinner voting is
  currently limited. Since Plurality is among the most basic,
  best-known rules, asking for its generalization to the multiwinner
  setting is a natural, fundamental issue.

\item There is a growing number of applications of multiwinner rules,
  many of which are much closer to the worlds of artificial
  intelligence, multiagent systems, and business, than to the world of
  politics. For example, Lu and
  Boutilier~\cite{bou-lu:c:chamberlin-courant,bou-lu:c:value-directed-cc}
  discuss several applications pertaining to business settings, for example,
  how a company can decide which set of items to advertise to their
  clients if the total number of items that they may advertise is
  limited.
Skowron~\cite{skow:multiwinner-models} studies applicability of various multiwinner rules to running different types of indirect elections. 
Elkind et
al.~\cite{elk-fal-sko-sli:c:multiwinner-properties} and Skowron et
al.~\cite{sko-fal-lan:c:collective} discuss further applications,
ranging from shortlisting candidates 
through various types of resource allocation tasks 
to choosing committees of representatives such as parliaments. 

\end{enumerate}
For the above two reasons, we believe that it is important to
understand fundamental properties of various (families of) multiwinner
rules. In this paper we explore and study one such fundamental
property and the corresponding family of rules.


\section{Preliminaries}\label{sec:prelim}

An election is a pair $E = (C,V)$, where $C = \{c_1, \ldots, c_m\}$ is
a set of candidates and $V = (v_1, \ldots, v_n)$ is a collection of
voters. 
Throughout the paper, we reserve the symbol $m$ to denote the number
of candidates.  Each voter $v_i$ is associated with a preference
order~$\pref_i$ in which $v_i$ ranks the candidates from his or her
most desirable one to his or her least desirable one. If $X$ and $Y$
are two (disjoint) subsets of $C$, then by $X \pref_i Y$ we mean that
for each $x \in X$ and each $y \in Y$ it holds that $x \pref_i y$. For
a positive integer $t$, we denote the set $\{1, \ldots, t\}$ by $[t]$.

\paragraph{Single-Winner Voting Rules.}
A single-winner voting rule~$\calR$ is a function that, given an
election $E = (C,V)$, outputs a subset $\calR(E)$ of candidates that
are called (tied) winners of this election. There is quite a variety
of single-winner voting rules, but in this paper it suffices to
consider scoring rules only.  Given a voter~$v$ and a candidate~$c$,
we write $\pos_v(c)$ to denote the position of $c$ in $v$'s preference
order (for example, if $v$ ranks $c$ first then $\pos_v(c) = 1$).  A
scoring function for $m$ candidates is a function $\gamma_m \colon [m]
\rightarrow \naturals$ such that for each $i \in [m-1]$ we have
$\gamma_m(i) \geq \gamma_m(i+1)$. Each family of scoring functions $\gamma=(\gamma_{m})_{m \in \naturals}$ (one function for each possible
choice of $m$) defines
a voting rule $\calR_\gamma$
as follows.  Let $E = (C,V)$ be an election with $m$
candidates. Under $\calR_\gamma$, each candidate
$c \in C$ receives $\score(c) := \sum_{v \in V}\gamma_m(\pos_v(c))$
points and the candidate with the highest number of points wins. (If
there are several such candidates, then they all tie as winners.)
We often refer to the value $\score(c)$ as the $\gamma$-score of $c$.

The following scoring functions and scoring rules are particularly
interesting.  The $t$-approval scoring function $\alpha_t$ is defined
as $\alpha_t(i) := 1$ for $i \leq t$ and $\alpha_t(i) := 0$
otherwise. (If $t$ is fixed, then the definition of $\alpha_t$ does
not depend on $m$. In such cases $\alpha_t$ can both be viewed as a
scoring function and as a family of scoring functions.)
%
%
For example, Plurality is $\calR_{\alpha_1}$, the
$t$-Approval rule is $\calR_{\alpha_t}$, and the Veto rule is
$\calR_{(\alpha_{m-1})_{m \in \naturals}}$.
The Borda scoring function (for $m$ candidates), $\beta_m$, is defined as 
%
  $\beta_m(i) := m - i$,
%
and $\calR_{\beta}$ is the Borda rule, where $\beta = (\beta_{m})_{m \in \naturals}$.

\paragraph{Multiwinner Voting Rules.}
A multiwinner voting rule $\calR$ is a function that, given an
election $E = (C,V)$ and a number~$k$ representing the size of the
desired committee, outputs a family $\calR(E,k)$ of size-$k$ subsets
of~$C$. The sets in this family are the committees that tie as
winners. 
As in the case of single-winner voting rules, one may need a
tie-breaking rule to get a unique winning committee, but we ignore
this aspect in the current paper.

We focus on the committee scoring rules, introduced by Elkind et
al.~\shortcite{elk-fal-sko-sli:c:multiwinner-properties}. Consider an
election $E = (C,V)$ and some committee $S$ of a given size $k$.  Let
$v$ be some voter in $V$. By $\pos_v(S)$ we mean the sequence $(i_1,
\ldots, i_k)$ that results from sorting the set $\{ \pos_v(c) \colon c
\in S\}$ in increasing order. For example, if $C = \{a, b, c, d, e\}$,
the preference order of $v$ is $a \succ b \succ c \succ d \succ e$,
and $S = \{a, c, d\}$, then $\pos_v(S) = (1, 3, 4)$.  If $I = (i_1,
\ldots, i_k)$ and $J = (j_1, \ldots, j_k)$ are two increasing
sequences of integers, then we say that $I$ \emph{(weakly) dominates} $J$
(denoted $I \succeq J$) if $i_t \leq j_t$ for each $t \in [k]$.  For
positive integers $m$ and $k$, $k \leq m$, by $[m]_k$ we mean the set
of all increasing size-$k$ sequences of integers from~$[m]$.

\begin{definition}\label{def:committee-scoring-function}
  [Elkind et al.~\shortcite{elk-fal-sko-sli:c:multiwinner-properties}]
  A committee scoring function for a multiwinner election with $m$
  candidates, where we seek a committee of size $k$, is a function $f_{m,k}
  \colon [m]_k \to \mathbb{N}$ such that for each two sequences $I,J
  \in [m]_k$ it holds that if $I \succeq J$ then $f(I) \geq f(J)$.
\end{definition}

Intuitively, the function $f_{m, k}$ from Definition~\ref{def:committee-scoring-function} assigns to each sequence $I$ of $k$ positions
the number of points that a committee $C$ gets from a voter $v$ when the members of $C$ stand on positions from $I$ in the preference order of $v$.

A committee scoring rule is defined by a family of committee scoring functions $f=(f_{m,k})_{k\le m}$, one function for each possible
choice of $m$ and $k$. Analogously to the case of single-winner scoring rules, we will denote such a multiwinner rule by $\calR_f$. Let $E = (C,V)$ be an election with $m$ candidates, let $k$, $k \leq
m$, be the size of the desired committee. Under the committee scoring rule $\calR_{f}$, every
committee $S \subseteq C$ with $|S|=k$ receives $\score(S) := \sum_{v
  \in V}f_{m, k}(\pos_v(S))$ points (for this notation, the election $E = (C,V)$
will always be clear from the context). The committee with the highest
score wins. (If there are several such committees, then they all tie
as winners.)

Many well-known multiwinner voting rules are, in
fact, families of committee scoring rules.
Consider the following
examples: 
%
\begin{enumerate}

\item SNTV, Bloc, and $k$-Borda rules pick $k$ candidates with the highest Plurality,
  $k$-Approval, and Borda scores, respectively. Thus, they are defined 
  through the following scoring functions:
  \begin{align*}
    f^\sntv_{m,k}\ \ (\row ik) &:= \textstyle\sum_{t=1}^{k}\alpha_1(i_t) = \alpha_1(i_1), \\
    f^\bloc_{m,k}\ \ \ \ \ (\row ik) &:= \textstyle\sum_{t=1}^{k}\alpha_k(i_t),\\
    f^\kborda_{m,k}(\row ik)& := \textstyle\sum_{t=1}^k\beta_m(i_t),
  \end{align*}
  respectively. Note that $f^\sntv_{m,k}$ is defined as a sum of
  functions that do not depend on either $m$ or $k$, $f^\bloc_{m,k}$
  is defined as a sum of functions that depend on $k$ but not $m$, and
  $f^\kborda_{m,k}$ is defined as a sum of functions that depend on
  $m$ but not $k$. 





\item The two versions of the Chamberlin--Courant rule that we
  consider are defined through the committee scoring functions:
  \begin{align*}
    f^{\beta\hbox{-}\cc}_{m,k}\ (\row ik) &:= \beta_m(i_1),\\
    f^{\alpha_k\hbox{-}\cc}_{m,k}(\row ik)&:=\alpha_k(i_1),
  \end{align*}
  respectively.  The first one defines the classical
  Chamberlin-Courant rule~\cite{cha-cou:j:cc} and the second one
  defines what we refer to as $k$-Approval Chamberlin--Courant rule
  (approval-based variants of the Chamberlin--Courant rule were
  introduced by Procaccia, Rosenschein, and
  Zohar~\cite{pro-ros-zoh:j:proportional-representation}, and were
  later studied, for example, by Betzler, Slinko, and
  Uhlman~\cite{bet-sli-uhl:j:mon-cc}, Aziz et
  al.~\shortcite{azi-bri-con-elk-fre-wal:c:justified-representation}
  and Skowron and Faliszewski~\cite{sko-fal:c:maxcover}). For brevity,
  we sometimes refer to $k$-Approval Chamberlin--Courant rule as the
  $\alpha_k${-}CC rule.
  
  Intuitively, under the Chamberlin--Courant rules, each voter is
  represented by the committee member that this voter ranks highest;
  the Chamberlin--Courant rule chooses a committee $C$ that maximizes
  the sum of the scores that the voters give to their representatives
  in $C$ (which characterizes the total satisfaction of the society
  with the assignment of representatives to
  voters). 

  



\end{enumerate}

\newcommand{\evote}[8]{ #1 \succ #2 \succ #3 \succ #4 \succ #5\succ #6\succ #7 \succ #8}

In the next example we show the differences between the above
rules. Since these rules are designed to satisfy different desiderata,
it turns out that on the same election they may provide significantly
different outcomes.

\begin{example}\label{example:1}
  Let us consider the set of candidates $C = \{a,b,c,d,e,f,g,h\}$ and
  eight voters with the following preference orders:
  \begin{align*}
     v_1 &\colon \evote a f c g h e b d , &
     v_2 &\colon \evote c e g h a f b d , \\
     v_3 &\colon \evote a f c h g e b d , &
     v_4 &\colon \evote d e h g a f b c , \\
     v_5 &\colon \evote b c g h a e f d , &
     v_6 &\colon \evote e g d h a b f c , \\
     v_7 &\colon \evote b d h g a e f c , &
     v_8 &\colon \evote f h d g a b e c.
  \end{align*}
Let the committee size $k$ be $2$.  It is easy to compute the
  winners under the SNTV and Bloc rules. For the former, the unique
  winning committee is $\{a,b\}$ (these are the only two candidates
  that are ranked in the top positions twice, and for the latter it is $\{e,f\}$
  (these are the only two candidates that are ranked among top two
  positions three times; all the other candidates are ranked there at
  most twice).  A somewhat tedious calculation shows that the unique
  $k$-Borda winning committee is $\{g,h\}$. (The Borda scores of the
  candidates $a, b, c, d, e, f, g, h$ are, respectively:
  \[
  32,\ 22,\ 23,\ 23,\ 28,\ 26,\ 35,\ 35.
  \]
  Finally, further calculations show that, under the (classical)
  Chamberlin--Courant rule, the unique winning committee is $\{c,d\}$.
  (While it is tedious to compute these results by hand, and indeed we
  used a computer to find them, the intuition for the $k$-Borda and
  Chamberlin--Courant winners is as follows: $g$ and $h$ are always
  ranked in the middle of each vote, or slightly above, so that they
  get high total Borda score, whereas $c$ and $d$ are ranked so that
  one of them is (almost) always ahead of $g$ and $h$, whereas the
  other is in the last position. This way, as representatives, $c$ and
  $d$ get higher scores than $g$ and $h$, even though their total
  Borda score is lower.)  Finally, it is relatively easy to verify
  that under $\alpha_k$-CC, the winning committee is $\{e,f\}$ (its
  $\alpha_k$-CC score is six; there is no other committee whose
  members are ranked among top two positions of six or more voters).
\end{example}

Following the nomenclature of Elkind et
al.~\shortcite{elk-fal-sko-sli:c:multiwinner-properties}, we say that
a committee scoring rule $\calR_f$ defined by the family $f=(f_{m,k})_{k\le m}$ 
of committee scoring functions  is \emph{weakly separable} if  there exists a family $(\gamma_{m,k})_{k\le m}$ 
of single-winner scoring functions, $\gamma_{m,k}\colon [m]_k\to \N$ such that:
\begin{align*}
    f_{m,k}(i_1, \ldots, i_k) = \sum_{t=1}^{k}\gamma_{m,k}(i_t).
\end{align*}
Intuitively, under a weakly separable rule, we can compute the scores
of all candidates separately and the rule picks up the $k$ candidates
with the highest scores.  $\calR_f$ is called \emph{separable} if for
fixed $m$ the function $\gamma_{m,k}$ does not depend on $k$. We see
that SNTV and $k$-Borda are separable, that Bloc is only weakly
separable, and that neither of our two versions of the
Chamberlin-Courant rule is weakly separable.  Elkind et
al.~\shortcite{elk-fal-sko-sli:c:multiwinner-properties} show that
separable rules have somewhat different properties than weakly
separable ones.

Our two variants of the Chamberlin--Courant rule are what Elkind et
al.~\cite{elk-fal-sko-sli:c:multiwinner-properties} call
\emph{representation focused} rules. A committee scoring rule
$\calR_f$, defined through a family $f = (f_{m,k})_{k \leq m}$ of
committee scoring functions, is representation focused if there exists
a family of single-winner scoring functions $\gamma =
(\gamma_{m,k})_{k \leq m}$ such that:
\[
  f_{m,k}(i_1, \ldots, i_k) = \gamma_{m,k}(i_1).
\]
Both Chamberlin--Courant and $\alpha_k$-CC are representation-focused,
but, somewhat surprisingly, so is SNTV.

The above two classes of rules (weakly separable ones and
representation-focused ones) can be further generalized. Recently,
Skowron, Faliszewski, and Lang~\shortcite{sko-fal-lan:c:collective}
introduced a new class of multiwinner rules based on OWA
operators\footnote{OWA stands for ``ordered weighted average''. OWA
  operators were introduced by Yager~\shortcite{yag:j:owa} in the
  context of multicriteria decision making. Kacprzyk et
  al.~\shortcite{kac-nur-zad:b:owa-social-choice} describe their
  applications in the context of collective choice.} (a variant of
this class was also studied by Aziz et
al.~\shortcite{azi-bri-con-elk-fre-wal:c:justified-representation,azi-gas-gud-mac-mat-wal:c:approval-multiwinner};
Elkind and Ismaili~\cite{conf/aldt/ElkindI15} consider a different
family of multiwinner rules defined through OWA operators as
well). While they did not directly consider elections based on
preference orders, we can apply their main ideas to committee scoring
rules.

An OWA operator $\Lambda$ of dimension $k$ is a sequence $\Lambda =
(\lambda^1, \ldots, \lambda^k)$ of nonnegative reals.\footnote{We
  slightly generalize the notion and, unlike
  Yager~\shortcite{yag:j:owa}, we do not require that $\lambda^1+
  \ldots+ \lambda^k=1$.}

%
%

\begin{definition}\label{def:owa-csr}
  Let $\Lambda=(\Lambda_{m, k})_{k \leq m}$ be a family of OWA
  operators such that $\Lambda_{m, k} = (\lambda^1_{m, k}$, $\ldots,
  \lambda^k_{m, k})$ has dimension $k$ (one size-$k$ vector for each
  pair $m, k$). Let $\gamma =(\gamma_{m,k})_{k\le m}$ be a family of
  (single-winner) scoring function (one scoring function for each pair
  $m, k$). Then $\gamma$ together with $\Lambda$ define a family of committee
  scoring functions $f = f_{m,k}(\Lambda,\gamma)$ such that for
  each $(i_1, \ldots, i_k) \in [m_k]$ we have:
\begin{align*}
    f_{m,k}(i_1, \ldots, i_k) = \sum_{t=1}^{k}\lambda^t_{m,k} \gamma_{m,k}(i_t).
\end{align*}
The committee scoring rule $\calR_f$ corresponding to the family $f$ 
is called {\em OWA-based}. 
\end{definition}

%

Intuitively, the OWA operators specify to what extent the voters care
about each  member of the committee, depending how this member is 
ranked among the other ones.
Indeed, every weakly separable rule is OWA-based with OWA operators of
the form $(1, \ldots, 1)$, which means that under weakly separable
rules the voters care about all the committee members equally.
%
%
Representation-focused rules are OWA-based through OWA operators of
the form $(1,0, \ldots, 0)$, which intuitively means that the voters
care about their top-ranked committee members only.
%
%
Another interesting group of OWA-based committee scoring rules is
defined through OWA operators of the form $(1, \frac{1}{2}, \ldots,
\frac{1}{k})$ and the $t$-Approval scoring functions (for some choice
of $t$).  We refer to these rules as $\alpha_t$-PAV (in essence, these are
variants of the Proportional Approval Voting rule, cast into the
framework of committee scoring rules by assuming that every voter
approves exactly his or her top $t$ candidates).  Intuitively, such
OWA operators indicate the decreasing attention the voters pay to
their lower ranked committee members.
%
For more discussion of the OWA-based rules, we refer the reader to the
works of Skowron et al.~\cite{sko-fal-lan:c:collective} and Aziz et
al.~\cite{azi-bri-con-elk-fre-wal:c:justified-representation,azi-gas-gud-mac-mat-wal:c:approval-multiwinner}
(the latter ones include a more detailed discussion of PAV; see also
the work of Kilgour~\cite{kil:chapter:approval-multiwinner} for a
description of this rule).


%

\begin{remark}
  We note that in most cases the OWA vectors $\Lambda_{m, k}$ used to
  define OWA-based rules do not depend on $m$. Yet, formally, we allow
  for such a dependency in order to build the relation between our
  general framework in which committee scoring functions $f_{m, k}$
  might depend on $m$ in any, even not very intuitive, way, and the
  world of OWA-based rules.
\end{remark}

\section{Fixed-Majority Consistent Rules}\label{sec:fmcr}

We now start our quest for finding committee scoring rules that can be
seen as multiwinner analogues of Plurality.  We begin by describing
the fixed-majority criterion that, in our view, encapsulates the idea
of closeness to Plurality.  Then, we provide a class of committee
scoring rules---the top-$k$-counting rules---that contains all the
rules which satisfy the fixed-majority criterion.  Finally, we provide
a complete characterization of fixed-majority voting rules within the
class of top-$k$-counting voting rules.

\subsection{Initial Remarks}

One of the features that distinguishes Plurality among all
other scoring rules is the fact that it satisfies the simple majority
criterion.

\begin{definition}
  A single-winner voting rule $\calR$ satisfies the {\em simple majority criterion}
  if, for every election $E = (C,V)$ where more than half of the voters
  rank some candidate $c$ first, it holds that $\calR(E) = \{c\}$.
\end{definition}

Importantly,
the simple majority criterion characterizes Plurality within the class of single-winner scoring rules.
The result is a part of folklore (we provide the proof for the sake of completeness).

\begin{proposition}
  Let $\gamma = (\gamma_m)_{m \in \naturals}$ be a family of single-winner scoring functions, such that the scoring rule 
  $\calR_\gamma$ satisfies the simple majority criterion. Then for each $m$ it holds that $\gamma_m(1) >
  \gamma_m(2) = \cdots = \gamma_m(m)$, and thus $\calR_\gamma$ coincides with Plurality.
\end{proposition}
\begin{proof}
  It is straightforward to verify that if for each $m$ we have
  $\gamma_m(1) > \gamma_m(2) = \cdots = \gamma_m(m)$ then
  $\calR_\gamma$ satisfies the simple majority criterion.  For the
  other direction, assume that $\calR_\gamma$ satisfies the simple
  majority criterion. This immediately implies that for each $m \geq
  2$ we have $\gamma_m(1) > \gamma_m(m)$ (otherwise all the candidates
  would always tie as winners). Hence for $m=2$ the result follows.
  
  Let us fix $m \geq 3$. For each positive integer $n$,
  define the election $E_n = (C,V_n)$ with the candidate set $C = \{c_1,
  \ldots, c_m\}$ and with $V_n$ containing:
  \begin{align*}
    &n+1 \text{ voters with preference order } c_1 \pref c_2 \pref \cdots \pref c_m \text{, and} \\
    &n \text{ voters with preference order } c_2 \pref c_3 \pref \cdots \pref c_m \pref c_1.
  \end{align*}
  Since $\calR_\gamma$ satisfies the simple majority criterion, it
  must be the case that $c_1$ is the unique $\calR_\gamma$-winner for
  each $E_n$.   Further, for a
  given value of $n$, the difference between the scores of $c_1$ and
  $c_2$ in $E_n$ is:
  \begin{align*}
     \score(c_1) - \score(c_2) &= \big((n+1)\gamma_m(1) + n \gamma_m(m) \big) 
                               - \big( (n+1)\gamma_m(2) + n\gamma_m(1) \big) \\
     &= \gamma_m(1) - \gamma_m(2) + n\big( \gamma_m(m) - \gamma_m(2)\big).
  \end{align*}
  Thus, if it held that $\gamma_m(2) > \gamma_m(m)$, then---for large
  enough value of $n$---candidate $c_1$ would not be a winner of
  $E_n$. This implies that $\gamma_m(2) = \cdots = \gamma_m(m)$. Since
  $\gamma_m(1) > \gamma_m(m)$, we reach the conclusion that $\gamma_m(1) >
  \gamma_m(2) = \cdots = \gamma_m(m)$.
\end{proof}



There are at least two ways of generalizing the simple majority
criterion to the multiwinner setting. We choose perhaps the simplest
one, the \emph{fixed-majority criterion} introduced by
Debord~\shortcite{deb:j:prudent} (other notions of majority studied by
Debord are variants of the Condorcet principle and are incompatible
with Plurality and scoring rules in general).

\begin{definition}
  A multiwinner voting rule~$\calR$ satisfies the fixed-majority
  criterion for $m$ candidates and committee size $k$ if, for every
  election $E = (C,V)$ with $m$ candidates, the following holds: if
  there is a committee $W$ of size $k$ such that more than half of the
  voters rank all the members of $W$ above the non-members of $W$,
  then $\calR(E,k) = \{W\}$.
  We say that $\calR$ satisfies the fixed-majority criterion if it
  satisfies it for all choices of $m$ and $k$ (with $k\le m$).
\end{definition}



\begin{remark}
Another way of extending the simple majority criterion to the multiwinner
case would be to say that, if a committee~$W$ is such that for each
$c \in W$ a majority of voters rank $c$ among their top $k$ positions (possibly a different
majority for each $c$), then $W$ must be a winning committee.
However, consider the following votes over the candidate set $\{a,b,c\}$:
%
\begin{align*}
  v_1 \colon a > b > c, \quad
  v_2 \colon a > c > b, \quad
  v_3 \colon b > c > a. \quad
\end{align*}
For $k=2$, all three committees, $\{a,b\}$, $\{a,c\}$, and $\{b,c\}$,
have majority support in the sense  just described. We feel that 
this is 
against the spirit of the 
simple majority criterion.
Thus, and since we have not found any other convincing ways of
generalizing the simple majority criterion to the multiwinner setting,
we focus on Debord's fixed-majority notion.
\end{remark}

One can verify that the Bloc rule satisfies the fixed-majority
criterion and that SNTV does not (it will also follow formally from
our further discussion).  This means that in the axiomatic sense, Bloc
is closer to Plurality than SNTV. This is quite interesting since
one's first idea of generalizing Plurality would likely be to think of
SNTV.  Yet, Bloc is certainly not the only committee scoring rule that
satisfies 
our criterion.  Let us consider the following rule.

\begin{definition}
  Let $k$ be the size of committee to be elected.  The
  Perfectionist rule is defined through the family $f$ of scoring function $f_{m,k}$
  such that $f_{m,k}(i_1, \ldots, i_k) = 1$, if $(i_1, \ldots, i_k) =
  (1, \ldots, k)$, and $f_{m,k}(i_1, \ldots, i_k) = 0$, otherwise. In
  other words, a voter gives score of 1 to a committee only if its
  members occupy the top $k$ positions of his or her vote.
\end{definition}

Alternatively, the Perfectionist rule can be viewed as an OWA-based
rule defined by the family $\Lambda =(\Lambda_{m, k})_{k \leq m}$ of
OWA operators, where $\Lambda_{m, k}=(0,0,\ldots,1)$, and the family
of $k$-Approval scoring function $\gamma_{m, k} = \alpha_k$.

\begin{example}\label{example:2}
  Let us, once again, consider the election from
  Example~\ref{example:1}.  In this election, for the committee size
  $k = 2$, Perfectionist assigns two points to committee $\{a,f\}$,
  one point to each of $\{b,c\}$, $\{b,d\}$, $\{c,e\}$, $\{d, e\}$,
  $\{e,g\}$, and $\{f,h\}$, and zero points to all the other
  committees. Thus, $\{a,f\}$ is the unique winning committee.
\end{example}

We note that Perfectionist satisfies the fixed-majority criterion and
that it closely resembles Plurality.  The following remark strongly
highlights this similarity.


\begin{remark}
Consider a situation where the voters extend their rankings of candidates to 
rankings of committees in some natural way.
Then, for each voter, the best committee would consist of his or her $k$ best candidates.
As a result,
running Plurality on the profile of preferences over the
committees would give the same result as running Perfectionist over the profile of preferences over the candidates.
\end{remark}

%
Naturally,
not all committee scoring rules satisfy the fixed-majority criterion.
For example,
neither $k$-Borda nor the
Chamberlin--Courant rule do.  
To see this, it suffices to note that
for $k = 1$ they both become the single-winner Borda rule, which fails
the simple majority criterion.

In what follows,
we will be interested in knowing which committee scoring rules do satisfy the fixed majority criterion and which do not.


\subsection{Top-$\boldsymbol{k}$-Counting Rules}

To characterize the committee scoring rules that satisfy the
fixed-majority criterion, we introduce a class of scoring
functions that depend only on the number of committee members ranked
in the top $k$ positions.

\begin{definition}\label{topkcounting}
We say that a committee scoring function $f_{m, k}\colon [m]_k\to\N$,  is 
  \emph{top-$k$-counting} if there is a function $g_{m, k} \colon \{0,
  \ldots, k\} \rightarrow \mathbb{N}$ such that $g_{m, k}(0)=0$ and for each $(i_1,
  \ldots, i_k) \in [m]_k$ we have $f_{m, k}(i_1, \ldots, i_k) = g_{m, k}( | \{ t \in
  [k] \colon i_t \leq k\} | )$.
We refer to $g_{m, k}$ as the \emph{counting function} for $f_{m,k}$.
%
%
%
We say that a committee scoring rule $\calR_f$ is
\emph{top-$k$-counting} if it can be defined through a family of
top-$k$-counting scoring functions $f=(f_{m,k})_{k\le m}$.
\end{definition}

Both Bloc and Perfectionist 
are top-$k$-counting rules. 
The former uses the linear counting function $g_{m, k}(x) = x$,
while the latter uses the counting function $g_{m, k}$ which is a
step-function: $g_{m, k}(x) = 0$ for $x<k$ and $g_{m, k}(k) = 1$.
Another example of a top-$k$-counting rule is the $\alpha_k$-CC rule,
which uses the counting function
$g_{m, k}$ such that $g_{m, k}(0) = 0$
and $g_{m, k}(x) = 1$ for all $x \in [k]$.\par\medskip

Top-$k$-counting rules have a number of interesting features.  First,
their counting functions have to be nondecreasing. Second, every
top-$k$-counting rule is OWA-based. Third, every committee scoring
rule that satisfies the fixed-majority criterion is
top-$k$-counting. We express these facts in the following two
propositions and in Theorem~\ref{pro:fix-maj-top-k}.  For the rest of
the paper we make the assumption that $m \geq 2k$;
this assumption is mostly technical as our arguments are greatly
simplified by the fact that we can form two disjoint
committees of size~$k$. Further, it is also quite natural: one could say that if
we were to choose a committee consisting of more than half of the
candidates, then perhaps we should rather be voting for who should
\emph{not} be in the elected committee.

\begin{proposition}\label{pro:g-nondec}
  Let  $m\ge 2k$ and let $f_{m,k}\colon [m]_k\to\N$ be a top-$k$-counting scoring function
  defined through a counting function $g_{m, k}$. Then, $g_{m, k}$ is nondecreasing.
\end{proposition}

\begin{proof}
%
  Let $t\in \{0,\ldots,k\}$. Consider the sequences $I_t =
  (1, \ldots,t,k+1, \ldots, k+(k-t))$ and $I_{t+1} = (1,
  \ldots,t+1,k+1, \ldots, k+(k-t-1))$ from $[m]_k$. (Note that we need $m\ge 2k$ for defining $I_0$.) Since $I_{t+1}
  \succeq I_{t}$, we have that $f_{m,k}(I_{t+1}) \geq f_{m,k}(I_t)$. By the
  definition, however, we have that $f_{m,k}(I_{t+1}) = g_{m,k}(t+1)$ and $f_{m,k}(I_t)
  = g_{m,k}(t)$.  Hence, $g_{m, k}(t+1) \geq g_{m, k}(t)$.
\end{proof}

Without the assumption that $m \geq 2k$,
Proposition~\ref{pro:g-nondec} would have to be phrased more
cautiously, and would speak only of the existence of nondecreasing
counting function. (For example, for $m=k$, the function $g_{m,k}$ could be
arbitrary.)

\newcommand{\textproowa}{
Every top-$k$-counting rule is OWA-based.}


\begin{proposition}\label{pro:owa}
  \textproowa
\end{proposition}

\begin{proof}
  Let us consider a top-$k$-counting rule $\calR_f$, where
  $f=(f_{m,k})_{k\le m}$ is the corresponding family of
  top-$k$-counting functions defined by a family of counting functions
  $(g_{m,k})_{k\leq m}$. Let us consider one function $f_{m,k}$ from
  this family. We know that $f_{m,k}\colon [m]_k\to\N$ is a
  top-$k$-counting scoring function defined through a counting
  function $g_{m,k}$ so that $f_{m,k}(i_1, \ldots, i_k) = g_{m,k}(s)$,
  where $s = |\{ t \in [k] \colon i_t \leq k\}|$. As $g_{m,k}(0)=0$,
  we have
\begin{align*}
f_{m,k}(i_1, \ldots i_k) &= g_{m,k}(s)-g_{m,k}(0)= \textstyle\sum_{t=1}^s (g_{m,k}(t)-g_{m,k}(t-1)) \\
                         &=\textstyle\sum_{t=1}^k \alpha_k(i_t)\cdot (g_{m,k}(t)-g_{m,k}(t-1)),
\end{align*}
from which we see that $\calR_f$ is OWA-based through the family of OWA operators:
\begin{align*}
\Lambda_{m,k} = (g_{m,k}(1)-g_{m,k}(0), g_{m,k}(2)-g_{m,k}(1), \ldots, g_{m,k}(k)-g_{m,k}(k-1)),
\end{align*}
and the family of $k$-Approval scoring functions ($\gamma_{m,k} = \alpha_k$).
\end{proof}

%
  %
  %
%

%
%

In the next theorem (and in many further theorems) we speak of a
committee scoring rule $\calR_f$ defined through a family of committee
scoring functions $f = (f_{m,k})_{2k\le m}$. We use this notation as a
shorthand for the assumption that the theorem is restricted to the
cases where $2k \leq m$.

\begin{theorem}\label{pro:fix-maj-top-k}
  Let $f=(f_{m,k})_{2k\le m}$ be a family of committee scoring functions. 
  If $\calR_f$ satisfies the fixed-majority criterion,  then $\calR_f$ is top-$k$-counting.
\end{theorem}

\begin{proof}
%
  Let us fix two numbers $m$ and $k$ such that $2k \leq m$.  Consider
  an election with $m$ candidates, where a committee of size $k$ is to
  be elected.  For each positive integer $t$ such that $0 \leq t \leq
  k$ we define the following two sequences from $[m]_k$:
  \begin{enumerate}

  \item $I_t = (1,\ldots,t, k+1, \ldots, k+k-t)$ is a sequence of
    positions of the candidates where the first $t$ candidates are
    ranked in the top $t$ positions and the remaining $k-t$ candidates
    are ranked just below the $k$th position.

  \item $J_t = (k-(t-1), \ldots, k, m-((k-t)-1), \ldots, m)$ is a
    sequence of positions where the first $t$ candidates are ranked
    just above (and including) the $k$th position, whereas the
    remaining $k-t$ candidates are ranked at the bottom.

  \end{enumerate}
  Among these, $I_k = (1, \ldots, k)$ is the highest-scoring sequence
  of positions and $J_k = (m-(k-1), \ldots, m)$ is the lowest-scoring
  sequence. Further, for every $t$ 
  we have $I_t \succeq J_t$ and, in effect, $f_{m,k}(I_t) \geq f_{m,k}(J_t)$.

  We claim that if there exists some $t \in \{0, \ldots, k\}$ such that
  $f_{m,k}(I_t) > f_{m,k}(J_t)$ then $\calR_f$ does not have the fixed-majority
  property. 
%
  For the sake of contradiction,
  assume that there is some $t$ such that
  $f_{m,k}(I_t) > f_{m,k}(J_t)$.
  Let $E = (C,V)$
  be an election with $m$~candidates and $2n + 1$~voters.
%
  The set of candidates is $C = X \cup Y \cup Z \cup D$, where 
  $X = \{x_1,     \ldots, x_t\}$,
  $Y = \{y_{t+1}, \ldots, y_k\}$,
  $Z = \{z_{t+1}, \ldots, z_k\}$,
  and $D$ is a set of sufficiently many dummy candidates so that
  $|C| = m$. 
  We focus on two committees, $M = X \cup Y$ and $N = X \cup Z$.
  The first $n + 1$ voters have preference order $X \pref Y \pref Z
  \pref D$, and the next $n$ voters have preference order $ Z
  \pref X \pref D \pref Y$.  Note that the fixed-majority criterion
  requires that $M$ be the unique winning committee.

  Committee $M$ receives the total score of $(n + 1) f_{m,k}(I_k) + n
  f_{m,k}(J_t)$, whereas committee $N$ receives the total score of $(n + 1)
  f_{m,k}(I_t) + n f_{m,k}(I_k)$.  The difference between these values is:
  \begin{align*}
    & (n + 1) f_{m,k}(I_k) + n f_{m,k}(J_t) - (n + 1) f_{m,k}(I_t) - n f_{m,k}(I_k) = \\
    & f_{m,k}(I_k) + n f_{m,k}(J_t) - (n + 1) f_{m,k}(I_t) = \\
    & f_{m,k}(I_k) - f_{m,k}(I_t) + n \big(f_{m,k}(J_t) - f_{m,k}(I_t)\big),
  \end{align*}
  which, for a large enough value of $n$, is negative (since, by assumption,
  we know that $f_{m,k}(J_t) < f_{m,k}(I_t)$ and so $f_{m,k}(J_t) - f_{m,k}(I_t)$ is negative).
  That is, for large enough $n$, committee $M$ does not win the
  election and $\calR_f$ fails the fixed-majority criterion.

  So, if $\calR_f$ satisfies the fixed-majority criterion, then for
  every $t \in \{0, \ldots, k\}$ we have that $f_{m,k}(I_t) =
  f_{m,k}(J_t)$. This, however, means that $f_{m,k}$ is a
  top-$k$-counting scoring function.  To see this, consider some
  sequence of positions $L = (\ell_1, \ldots, \ell_k)\in [m]_k$ where
  exactly the first $t$ entries are smaller than or equal to
  $k$. Clearly, we have that $I_t \succeq L \succeq J_t$ and so
  $f_{m,k}(I_t) = f_{m,k}(L) = f_{m,k}(J_t)$, which means that
  $f_{m,k}(i_1, \ldots, i_k)$ depends only on the cardinality of the
  set $\{ t \in [k] \colon i_t \leq k\}$. Since $m$ and $k$ where
  chosen arbitrarily (with $2k \leq m$), this completes the proof.
\end{proof}

Unfortunately, the converse of Theorem~\ref{pro:fix-maj-top-k}
does not hold: $\alpha_k$-CC, for example, is a top-$k$-counting rule that fails the
fixed-majority criterion.


\begin{example}\label{example:cc}
  Consider an election $E = (C,V)$ with $C = \{a,b,c,d\}$, $V = (v_1, v_2, v_3)$, and $k = 2$. Let the preference orders of the voters
  be:
  \begin{align*}
    v_1 &\colon \fourvote abcd,& 
    v_2 &\colon \fourvote abcd, & 
    v_3 &\colon \fourvote cdab.
  \end{align*}
  The fixed-majority criterion requires $\{a,b\}$ to be the only
  winning committee, while under $\alpha_k$-CC, other committees,
  such as $\{a,c\}$, have strictly higher scores. (Incidentally, this
  example also witnesses that SNTV fails the fixed-majority criterion;
  the fact is hardly surprising since SNTV is not a top-$k$-counting rule.)
\end{example}

\subsection{Criterion for Fixed-Majority Consistency}

In this section, we provide a formal characterization of those top-$k$-counting
rules that satisfy the fixed-majority criterion.  Together with
Theorem~\ref{pro:fix-maj-top-k}, this gives a full
characterization of committee scoring rules with this property.



\newcommand{\textthmcharacterization}{
 Let $f=(f_{m,k})_{2k\le m}$ be a family of committee scoring functions with the corresponding family $(g_{m, k})_{2k \leq m}$ of counting functions.   
 Then, 
  $\calR_f$ satisfies the fixed-majority criterion if and only if for every $k, m\in \N$, $2k \leq m$:
\begin{enumerate}
\item[(i)] $g_{m, k}$ is not constant, and 
\item[(ii)] for each pair of nonnegative integers
  $k_1,k_2$ with $k_1+k_2 \leq k$, it holds that:
  \begin{align*}
  g_{m,k}(k) - g_{m,k}(k-k_2) \geq g_{m,k}(k_1+k_2) - g_{m,k}(k_1).
  \end{align*}
  \end{enumerate}
}

\begin{theorem}\label{thm:characterization}
\textthmcharacterization
\end{theorem}

Condition (ii) in Theorem~\ref{thm:characterization} is a relaxation of the convexity
property for function $g_{m,k}$ and is illustrated in
Figure~\ref{fig:characterization}. We discuss this in more detail
after the proof of the theorem.

\begin{figure}
  \begin{center}
    \scalebox{0.85}{ 
  \includegraphics{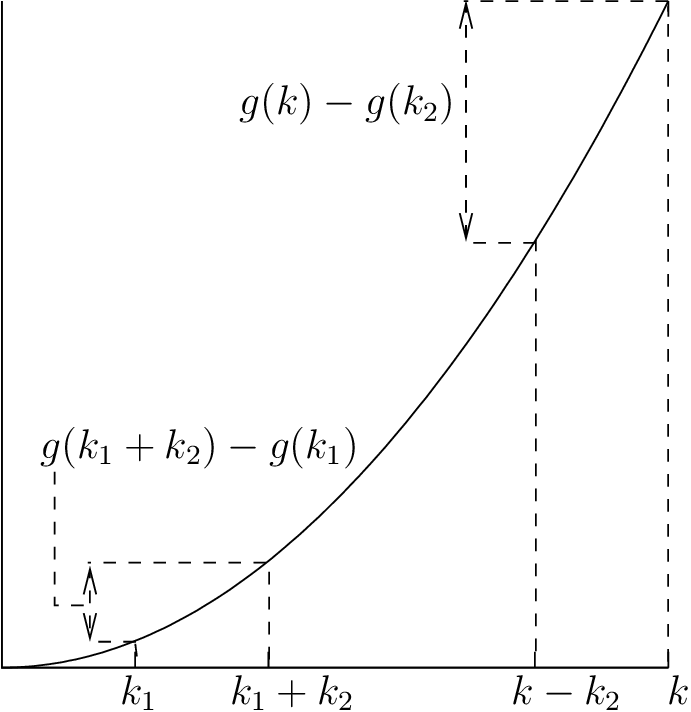}
  }
    \caption{\label{fig:characterization}Illustration of 
      the condition from
      Theorem~\ref{thm:characterization}.}
  \end{center}
\end{figure}

\begin{proof}[Proof of Theorem~\ref{thm:characterization}]
Let $f_{m,k}$ be one of the committee scoring functions and $g_{m,k}$ be its corresponding counting function.  By
  Proposition~\ref{pro:g-nondec}, $g_{m,k}$ is nondecreasing so the fact
  that it is non-constant is equivalent to $g_{m,k}(k)>g_{m,k}(0)$.  Moreover, we
  note that conditions (i) and (ii) imply that for each $k'$ with $0 \leq
  k' \leq k-1$, we have $g_{m,k}(k) > g_{m,k}(k')$.  To see this we take $k_2 =
  1$ and note that for each $k_1$ it holds that $g_{m,k}(k) - g_{m,k}(k - 1) \geq
  g_{m,k}(k_1 + 1) - g_{m,k}(k_1)$.  As $g_{m,k}(k)>g_{m,k}(0)$, for some $k_1$ we have that
  $g_{m,k}(k_1 + 1) - g_{m,k}(k_1) > 0$. Thus, we have that $g_{m,k}(k) - g_{m,k}(k-1) \geq
  g_{m,k}(k_1 + 1) - g_{m,k}(k_1) > 0$. Since $g_{m,k}$ is nondecreasing, it is also
  true that $g_{m,k}(k) > g_{m,k}(k')$.

  Let us now show that if for each $m$ and $k$, $g_{m,k}$ satisfies (ii) then $\calR_f$ has the
  fixed-majority property.
%
%
  Let $E = (C,V)$ be an election with $n$ voters and $m$ candidates,
  for which there is a size-$k$ committee $M$ such that a majority of
  the voters rank all members of $M$ in the top $k$ positions, but $M$
  loses to some committee $S \neq M$ (also of size~$k$).  That is, we
  have $\score(S) \geq \score(M)$.  Let $\xi$ be a rational number,
  $\frac{1}{2} < \xi \leq 1$, such that exactly $\xi n$ voters
  rank all the members of $M$ in the top $k$ positions; we will refer
  to these voters as $M$-voters and to the others as non-$M$-voters.

  Without loss of generality, we can assume that all the
  non-$M$-voters have identical preference orders. Indeed, if it were
  the case that $f_{m,k}(\pos_{v_i}(S)) - f_{m,k}(\pos_{v_i}(M)) >
  f_{m,k}(\pos_{v_j}(S)) - f_{m,k}(\pos_{v_j}(M))$ for some two non-$M$-voters
  $v_i$ and $v_j$, then we could replace the preference order of $v_j$
  with that of $v_i$ and increase the advantage of $S$ over $M$. If for
  all non-$M$-voters this difference were the same, then we could
  simply pick the preference order of one of them and assign it to all
  the other ones.

  Let $k_1$, $k_2$, $k_3$, and $k_4$ be four numbers such that:
  \begin{enumerate}
  
  \item $k_1$ is the number of candidates from $S \cap M$ that the non-$M$-voters
    rank among their top $k$ positions, 
  
  \item $k_2$ is the number of candidates from $S \setminus M$ that
    the non-$M$-voters rank among their top $k$ positions, and
  
  \item $k_3$ is the number of candidates from $C \setminus (S \cup
    M)$ that the non-$M$-voters rank among their top $k$ positions.
  
  \item $k_4$ is the number of candidates from $M \setminus S$ that
    the non-$M$-voters rank among their top $k$ positions.

  \end{enumerate}
  Without loss of generality, we can assume that $k_4 = 0$ and that
  $|S \setminus M| = k_2$ (since $m \geq 2k$, we can replace all
  members of $M \setminus S$ with candidates from $C \setminus M$,
  and, similarly, we can ensure that all members of $S \setminus M$
  are ranked among top $k$ positions by non-$M$-voters; these changes
  never decrease the score of $S$ relative to that of $M$).  In
  effect, we have that $k_1 + k_2 + k_3 = k$ and, since $|S\cap
  M|+|S\setminus M|=k$, we have that $|S\cap M|=k-k_2$.  We can assume
  that $k_2 > 0$ as otherwise we would have $S = M$. Given this
  notation, the difference between the scores of $M$ and $S$ is:
  %
  %
  \begin{align*}
    &\score(M) - \score(S) =\\
    &\xi n \cdot  g_{m,k}(k) + (1-\xi)n \cdot g_{m,k}(k_1)
             - \xi n \cdot g_{m,k}(k - k_2) -(1-\xi)n \cdot g_{m,k}(k_1+k_2) = \\
             & \xi n \cdot \big(g_{m,k}(k) - g_{m,k}(k-k_2)\big)
             - (1-\xi)n \cdot \big( g_{m,k}(k_1+k_2) - g_{m,k}(k_1) \big) > 0,
  \end{align*}
  where the second equality holds due to rearranging of terms, and the final
  inequality is an immediate consequence of the assumptions regarding
  the value of $\xi$ and the properties of $g_{m,k}$ (namely, that $g_{m,k}(k) -
  g_{m,k}(k-k_2) \geq g_{m,k}(k_1+k_2) - g_{m,k}(k_1)$ and that $g_{m,k}(k) - g_{m,k}(k-k_2) > 0$).
  This, however, contradicts the assumption that $\score(S) \geq
  \score(M)$ and, so, $\calR_f$ satisfies the fixed-majority
  criterion.\par\smallskip

  We now consider the other direction.  For the sake of contradiction,
  let us assume that $\calR_f$ satisfies the fixed-majority criterion
  but that there exist $m$ and $k$ such that either condition (i) or condition (ii) is not satisfied.  If $g_{m,k}$ is a
  constant function then $\calR_f$ fails the fixed-majority criterion
  because it always outputs all the subsets of size~$k$, independently of voters' preferences.  Thus we
  assume that $g_{m,k}$ is not constant. Suppose (ii) does not hold and
  there exist $k_1$ and $k_2$ with $k_1+k_2\le k$ such that $g_{m,k}(k) -
  g_{m,k}(k-k_2) < g_{m,k}(k_1+k_2)-g_{m,k}(k_1)$. We form an election with $m$
  candidates, $c_1, \ldots, c_m$, and $2n+1$ voters (we describe the
  choice of $n$ later). The first $n+1$ voters have preference order:
   \[
     c_1 \pref c_2 \pref \cdots \pref c_m,
   \]
   and the remaining $n$ voters have preference order:
   \[
      c_1 \pref \cdots \pref c_{k_1} \pref c_m \pref c_{m-1} \pref 
\cdots \pref c_{k_1+1}.
   \]
   Since $\calR_f$ satisfies the fixed-majority criterion, in this
   election it outputs the unique winning committee $M = \{c_1,
   \ldots, c_k\}$. However, consider committee $S$:
   \begin{align*}
   S = \{c_1, \ldots,
   c_{k_1+k_2}, c_{m}, \ldots, c_{m-(k-k_1-k_2)+1}\}.
   \end{align*}
   Since $m \ge 2k$, the difference between the scores of $M$ and $S$ is:
   \begin{align*}
     & \score(M) - \score(S) = \\
     &(n+1) g_{m,k}(k) + n g_{m,k}(k_1) 
               - (n+1) g_{m,k}(k_1+k_2) - n g_{m,k}(k-k_2) = \\
     & n\big( g_{m,k}(k) - g_{m,k}(k-k_2)\big)  + g_{m,k}(k) 
      - n\big( g_{m,k}(k_1+k_2) - g_{m,k}(k_1) \big) -g_{m,k}(k_1+k_2).
   \end{align*}
   Since $g_{m,k}(k) - g_{m,k}(k-k_2) < g_{m,k}(k_1+k_2)-g_{m,k}(k_1)$, we observe that for
   large enough $n$ the difference $\score(M) - \score(S)$ becomes
   negative.  This is a contradiction showing that (ii) holds.
\end{proof}

Let us take a step back and consider what condition (ii) from
Theorem~\ref{thm:characterization} means (recall
Figure~\ref{fig:characterization}).  Intuitively, it resembles the
convexity condition, but `focused' on $g_{m,k}(k)$.

\begin{definition}
Let $g_{m,k}$ be a counting function for some top-$k$-counting 
function $f_{m,k}\colon [m]_k\to\N$. We say that $g_{m,k}$ is {\em convex} if for each $k'$ such that $2 \leq k'
\leq k$, it holds that:
\begin{align*}
g_{m,k}(k') - g_{m,k}(k'-1) \geq g_{m,k}(k'-1) - g_{m,k}(k'-2).
\end{align*}
On the other hand, we say that $g$ is {\em concave} if for each $k'$
with $2 \leq k' \leq k$ it holds that:
\begin{align*}
g_{m,k}(k') - g_{m,k}(k'-1) \leq g_{m,k}(k'-1) - g_{m,k}(k'-2).
\end{align*}
\end{definition}

The notions of convexity and concavity are standard, but allow us
to express many features of top-$k$-counting rules in a very intuitive
way. For example, the following corollary is an immediate consequence
of Theorem~\ref{thm:characterization}.

\begin{corollary}\label{cor:characterization}
  Let $f=(f_{m,k})_{2k\le m}$ be a family of top-$k$-counting
  committee scoring functions with the corresponding family
  $(g_{m,k})_{2k\le m}$.  of counting functions.  The following hold:
  \begin{enumerate}
\item[(1)] if $g_{m,k}$ are convex, then $\calR_f$ satisfies the
  fixed-majority criterion, 
  and
\item[(2)] if $g_{m,k}$ are concave but not linear (that
  is, $\calR_f$ is not Bloc) then $\calR_f$ fails the fixed-majority
  criterion.
  \end{enumerate}
\end{corollary}

The counting function for the Bloc rule is linear (and, thus, both convex and
concave), and the counting function for the Perfectionist rule is convex, so
these two rules satisfy the fixed-majority criterion. On the other
hand, the counting function for $\alpha_k$-CC is concave and, so, this
rule fails the criterion (as we observed in Example~\ref{example:cc}). 

By Proposition~\ref{pro:owa}, a family of concave counting functions
$g_{m,k}$ corresponds to a nonincreasing OWA operator, and a family of
convex counting functions corresponds to a nondecreasing one.  Skowron
et al.~\shortcite{sko-fal-lan:c:collective} provided evidence that
rules based on nonincreasing OWA operators are computationally easier
than those based on general OWA operators (though, still their winners
tend to be $\np$-hard to compute). In the next section we show that
this seems to be the case for top-$k$-counting rules as well, but we
also provide a striking example highlighting certain dissimilarity.

\section{Complexity of Top-$\boldsymbol{k}$-Counting Rules}\label{sec:ctkcr}

In this section, we consider the computational complexity of
winner determination for top-$k$-counting rules which are based on either convex or
concave counting functions. We start by considering several examples.

It is well-known that Bloc winners can be computed in polynomial
time. The same holds for the Perfectionist rule.

\begin{proposition}\label{pro:perfectionist}
  Both the Bloc and Perfectionist winners are computable in
  polynomial time.
\end{proposition}

\begin{proof}
  The case of Bloc is well-known
 (as we mentioned, the Bloc rule is weakly separable and to form a winning committee of
  size $k$ it suffices to pick $k$ candidates with the highest $k$-Approval scores).
  To find a size-$k$ winning committee under the Perfectionist rule,
  for each voter $v$ we consider the set of her top-$k$ candidates as
  a committee, and compute the score of that committee in the
  election.  We output those committees---among the considered
  ones---that have the highest score.  Correctness follows by noting
  that the committees that the algorithm considers are the only ones
  with nonzero scores.
\end{proof}

While the result for the Perfectionist rule is very simple, it stands
in sharp contrast to the results of Skowron, Faliszewski and
Lang~\shortcite{sko-fal-lan:c:collective}.  By
Proposition~\ref{pro:owa}, Perfectionist is defined through the OWA
operator $(0, \ldots, 0,1)$, and Skowron et al. have shown that, in
general, rules defined through this operator are $\np$-hard to compute
and very difficult to approximate. However, their result relies on the
fact that voters can approve any number of candidates, while in our
case they have to approve exactly $k$ of them. This shows very clearly
that even though top-$k$-counting rules are OWA-based, we cannot
simply carry over the hardness results of Skowron et
al.~\shortcite{sko-fal-lan:c:collective} or Aziz et
al.~\shortcite{azi-gas-gud-mac-mat-wal:c:approval-multiwinner} to our
framework.

\newcommand{\singularity}[1]{\mathrm{sing}(#1)}

We can generalize Proposition~\ref{pro:perfectionist} to rules that
are, in some sense, similar to Perfectionist.  To this end, and to
facilitate our later discussion regarding the complexity of
top-$k$-counting rules, we define the following property of counting functions.


\begin{definition}
  Let $g_{m,k}$ be a counting function for a top-$k$-counting function
  $f_{m,k}\colon [m]_k\to\N$.  We define the singularity of $g_{m,k}$,
  denoted $\singularity{g_{m,k}}$, to be
$$
  \singularity{g_{m,k}} = \argmin_{2 \leq i \leq k} \big( g_{m,k}(i) - g_{m,k}(i - 1) \neq g_{m,k}(i - 1) - g_{m,k}(i - 2) \big).
$$
\end{definition}

Loosely speaking, $\singularity{g_{m,k}}$ is the smallest integer in $\{2,
\ldots, k\}$ for which the differential of $g_{m,k}$ changes. For Bloc
 (which is an exception) we define $\singularity{g_{m,k}}$ to be
$\infty$, since the differential is a constant function. Naturally,
for all other non-constant rules, the singularity is finite. For
example, for Perfectionist we have $\singularity{g_{m,k}} = k$.

We generalize the polynomial-time algorithm for Perfectionist 
to similar rules, for which the value $\singularity{g_{m,k}}$ is close to
$k$.  

%


  
  
\newcommand{\textpronearlyperfectionist} {Let $f=(f_{m,k})_{2k\le m}$
  be a top-$k$-counting family of (polynomial-time computable)
  committee scoring functions with the corresponding family of
  counting functions $(g_{m,k})_{2k\le m}$, and $\calR_f$ be the
  corresponding top-$k$-counting rule. Let $q$ be a constant, positive
  integer such that $k - \singularity{g_{m,k}} \leq q$ holds for all
  $m$ and $k$. Then $\calR_f$ has a polynomial-time computable winner
  determination problem.}


\begin{proposition}\label{pro:nearly-perfectionist}
\textpronearlyperfectionist
\end{proposition}

\begin{proof}
  Let the input consist of election $E = (C,V)$ and positive integer
  $k$, and let $W$ be a winning committee in $\calR(E,k)$. We assume
  that $q < \frac{k}{2}$ (if it were not the case, then $k\le 2q$
  would be small and we could solve the problem using brute-force). We
  consider two cases: (1) there is at least one voter that has at
  least $\singularity{g_{m,k}}$ of his or her top $k$ candidates in
  $W$; (2) every voter has less than $\singularity{g_{m,k}}$ of his or
  her top $k$ candidates in $W$.

  If case (1) holds, then we can compute $W$ (or some other winning
  committee) by checking, for each voter $v$, all the committees that
  consist of at least $\singularity{g_{m,k}}$ candidates that $v$ ranks
  among his or her top $k$ positions. Since $k - \singularity{g_{m,k}} \leq q$, the
  number of committees that we have to check for each voter is:
  \[
  \sum_{t = \singularity{g_{m,k}}}^k {k \choose t} {m \choose k - t} \leq (q+1)
  \cdot {k \choose k - \singularity{g_{m,k}}} {m \choose k -
    \singularity{g_{m,k}}},
  \]
  which is a polynomial in $k$ and $m$. The above inequality requires
  some care: We have that $\singularity{g_{m,k}} > \frac{k}{2}$
  (because $k - \singularity{g_{m,k}} \leq q < \frac{k}{2}$) and, in
  effect, we have that for each $t \in \{\singularity{g_{m,k}},
  \ldots, k\}$ it holds that ${k \choose t} = {k \choose k - t} \leq
  {k \choose k - \singularity{g_{m,k}}}$ and ${m \choose k - t} \leq
  {m \choose k -\singularity{g_{m,k}}}$.

  If case (2) holds, then from the fact that $g_{m,k}(x) - g_{m,k}(x - 1)$ is a
  constant for $x \leq \singularity{g_{m,k}}$, we infer that $g_{m,k}(x)$ is
  effectively linear.
  Then, it suffices to compute the winning committee using the Bloc rule.
%
While we do not know which of the two cases holds, we can compute
the  two committees, one as in case (1) and one as in case (2), and
  output the one with the higher score
  (or either of them, in case of a tie).
  %
  %
%
\end{proof}

\begin{example}
  Consider the following committee scoring function:
  \[
     f_{m,k}'(i_1, \ldots, i_k) = f_\bloc(i_1, \ldots, i_k) + f_\perf(i_1, \ldots, i_k) = \alpha_k(i_1) + \cdots + \alpha_k(i_{k-1}) + 2\alpha_k(i_k).
  \]
  As a simple application of Proposition~\ref{pro:nearly-perfectionist},
  we get that the committee scoring rule $\calR_{f'}$ defined through $f'$ is polynomial-time computable.
  This rule can be seen as a variant of
  Bloc, where a voter gives additional one bonus point to a committee if
  he or she approves of all its members. By
  Corollary~\ref{cor:characterization}, this rule is fixed-majority
  consistent.

  It is also interesting to consider the rule which is defined through
  the following committee scoring function:
  \[ 
     f_{m,k}''(i_1, \ldots, i_k) = f_\sntv(i_1, \ldots, i_k) + f_\perf(i_1, \ldots, i_k) = \alpha_1(i_1) + \alpha_k(i_k).
  \]
  The corresponding rule is also polynomial-time computable (it
  suffices to compute an SNTV winning committee, and compare it with such committees whose all members stand on first $k$ positions in some voter's preference ranking), but it is not a top-$k$-counting
  rule and, so, it fails the fixed-majority criterion.
\end{example}

Yet, as one might expect, not all top-$k$-counting rules are
polynomial-time solvable and, indeed, most of them are not (under
standard complexity-theoretic assumptions). For example, $\alpha_k$-CC
is $\np$-hard (this follows quite easily from Theorem~1 of Procaccia
et al.~\shortcite{pro-ros-zoh:j:proportional-representation}; we
include a brief proof to substantiate the discussion and give the
reader some intuition).

\newcommand{\textprocccomplexity}{For
  $\alpha_k$-CC it  is $\np$-hard to decide whether or not there exists a committee with at least a given score
  (recall that $k$ in $\alpha_k$-CC is the committee size and, thus,  is part of the input).
}

\begin{proposition}\label{pro:cc-complexity}
  \textprocccomplexity
\end{proposition}

\begin{proof}
  The $\np$-hardness follows easily from a standard reduction from the
  \textsc{Exact Cover by 3-Sets} problem, abbreviated as \textsc{X3C}.
  In an instance of \textsc{X3C} we are given a family of $m$ subsets,
  $S_1,\ldots,S_{m}$, each of cardinality 3, chosen from a given
  universal set $U = \{x_1,\ldots,x_{3n}\}$; we ask if there are $n$
  subsets from the family whose union is $U$. Additionally, we may assume that each
  element of $U$ belongs to at most three subsets since it is
  well-known that this variant of \textsc{X3C} remains $\np$-complete.

  Given an instance of \textsc{X3C}, we create a candidate for each
  subset, and a voter for each element of $U$. Voters rank the
  elements of the subsets to which they belong in their top positions,
  then they rank some $n$ dummy candidates (different ones for each
  voter), and then all the remaining candidates (in some arbitrary,
  easy to compute, order). We ask for a committee of size $k = n$.  There
  is a winning committee with score $3n$ if and only if the answer for
  the input instance is ``yes.''
\end{proof}



We generalize the above $\np$-hardness result to the case of convex
top-$k$-counting rules $\calR_f$ for which there is some constant $c$ such that
for each $k$ and $m$ it holds that $k - \singularity{g_{m, k}} \geq k / c$ (that is, to the case of convex
counting functions for which the differential changes `early').  An
analogous result for concave counting functions
follows from the works of Skowron, Faliszewski, and
Lang~\shortcite{sko-fal-lan:c:collective} and Aziz et
al.~\shortcite{azi-gas-gud-mac-mat-wal:c:approval-multiwinner}.

\newcommand{\textthmconvexhard}{Let $\calR_f$ be a top-$k$-counting rule
  defined through a family $f$ of top-$k$-counting functions
  $f_{m,k}\colon [m]_k\to \N$ with the corresponding family of counting functions
  $(g_{m,k})_{k \leq m}$ that do not depend on $m$, $g_{m,k} = g_k$, and such that: 
%
  \begin{enumerate}


  \item 
    For each $x$, $0 \leq x \leq k$, $g_{k}(x)$ is
    computable in polynomial time with respect to $k$ (that is, there is a polynomial time algorithm
    that given $x$ and $k$ outputs $g_{k}(x)$). 
    Moreover, for  each $k$, 
    $g_{k}(k)$ is polynomially bounded in $k$.

  \item 
    There is a constant $c$ such that, for each size of committee  $k$
    greater than some fixed constant $k_0$, $g_{k}$ is convex and
    $k - \singularity{g_{k}} \geq k / c$.

  \end{enumerate}
  Then, deciding if there is a committee with at least a given score
  is $\np$-hard for $\calR_f$.}

\begin{theorem}\label{thm:convex-hard}
  \textthmconvexhard
\end{theorem}

\begin{proof}
  We prove $\np$-hardness of the problem 
  by giving a reduction from the \textsc{Clique} problem on regular
  graphs.  A graph is regular if all its vertices have the same
  degree.  In the \textsc{Clique} problem we are given a
  graph $G$ and an integer $h$, and we ask if there exists a set of
  $h$ pairwise adjacent vertices in $G$ (such a set of vertices is
  referred to as a \emph{size-$h$ clique}).  The problem remains
  $\np$-complete when restricted to regular
  graphs~\cite{gar-joh:b:int}.

  Let $G$ be the input regular graph, let $h$ be the size of the
  clique sought for, and let $\delta$ be the common degree of $G$'s
  vertices. If $h > \delta + 1$, then, of course, the graph does not
  contain a size-$h$ clique and we output a fixed ``no''-instance of
  our problem. Otherwise, we output an instance according to the
  following construction (intuitively, since each $g_k$ is convex, the rule
  promotes situations where voters rank many members of the committee
  among their top $k$ candidates; we exploit this fact).

  We set the committee size $k$ to be $(c+2)h$.  Since $g_k$ does not
  depend on the number of candidates in the election, this fixes the
  counting function that we work with and we will denote it $g$. If
  $k\le k_0$ (recall that $k_0$ is defined in the statement of the
  theorem), then we solve the input instance using brute force in
  polynomial time and output either a fixed ``yes''-instance or a
  fixed ``no''-instance, depending on the result.
  We note that for each $i$, $1 \leq i \leq \singularity{g}$, all the
  values $g(i) - g(i-1)$ are equal and, without loss of generality, we
  can assume them to either all be $0$s or all be $1$s (if this were
  not the case, we could scale $g$ appropriately). Similarly, since
  $g$ is convex, we can assume that
  $g(\singularity{g})-g(\singularity{g}-1) > 1$.  We note that $k -
  \singularity{g} \geq k / c = (c+2)h / c > h$ and, so,
  $\singularity{g} < k - h$.

  We form an election with the following candidates: 
%
  \begin{enumerate}

  \item For each vertex $v$ from the graph $G$, we create a candidate~$v$.


  \item We create a set \{$c_1, \ldots, c_{\singularity{g} - 2}\}$ of
    candidates, called the \emph{edge-filler candidates}.  These candidates will be in the top-$k$ positions of all
    the voters, and hence will be chosen to every  winning committee.

  \item We create a set $\{b_1, \ldots, b_{k-h-(\singularity{g}-2)}\}$
    of candidates, called \emph{general-filler candidates}. There will be sufficiently many voters who rank them in their top-$k$
    positions so that they will also be in every winning committee.
    
 \item We also create a set of dummy candidates, such that each dummy candidate is ranked among
  the top-$k$ positions of exactly one voter.
  
  \end{enumerate}
  Let $m$ be the total number of edges in $G$.  For each edge $e$, we
  create a set of $2g(k)$ voters corresponding to this edge; each
  voter in this set has the following candidates in the top $k$
  positions of his or her preference order:
  \begin{enumerate}

  \item The two candidates corresponding to the endpoints of $e$.

  \item All the edge-filler candidates.

  \item Sufficiently many dummy candidates (such that they are ranked
    among top $k$ positions only by this voter).

  \end{enumerate}
  Further, we create $2g(k) \cdot (m+h) \cdot g(k)$ filler voters, who rank the
  following candidates in the top $k$ positions:
  \begin{enumerate}

  \item All the edge-filler candidates.  

  \item All the general-filler candidates.

  \item Sufficiently many dummy candidates (different dummy candidates
    for each filler voter).
  \end{enumerate}

\noindent
  (The role of the $2g(k)$ multiplicity factor regarding both the edge
  voters and the filler voters it to ensure that the best committee
  does not contain any of the dummy candidates; this will become clear
  later in the proof.) 

  We ask whether there is a committee $W$ whose score is at least $T =
  T_1 + T_2 + T_3 + T_4$, where:
  \begin{align*}
    T_1 & = 2g(k) \cdot (m+h)\cdot g(k)\cdot g(k-h), \\
    T_2 &=  2g(k) \cdot m \cdot g(\singularity{g}-2), \\
    T_3 &=  2g(k) \cdot \delta h \cdot \big( g(\singularity{g}-1) - g(\singularity{g}-2) \big ),\\
    T_4 &=  2g(k) \cdot {\textstyle\binom{h}{2}} \big(  g(\singularity{g}) - g(\singularity{g}-2) -2( g(\singularity{g}-1) - g(\singularity{g}-2) )  \big).
  \end{align*}
 Note that each $T_i$, $1 \leq i \leq 4$, is nonnegative (for $T_4$ this is due
  to convexity of $g$).   The meaning of these values will become clear throughout the proof.
%
  This finishes the construction. Due to the
  assumptions regarding the counting function, the reduction is
  polynomial-time computable.

  Let us now argue that the reduction is correct. First, we claim that
  if a committee $W$ has a score of at least $T$, then it must contain
  all the edge-filler candidates and all the general-filler
  candidates. We note that altogether we have $k-h$ edge-filler and
  general-filler candidates.  Consider some committee $W'$ that
  contains $k - h - x$ candidates of these two types, where $x \geq
  1$. This means that $W'$ contains at most $h+x$ dummy candidates.
 
 
  Let $y$ be the number of filler voters that rank at least $k-h$
  members of $W'$ among their top $k$ positions. Let us call these
  filler voters well-satisfied. For each of the well-satisfied filler
  voters, the members of $W'$ ranked on top $k$ positions are
  (a)~the $k-h-x$ edge-filler and general-filler candidates from $W'$,
  and
  (b)~at least $x$ unique dummy candidates.
  Thus it must hold that $xy \leq h+x$ and, so, $y \leq
  \frac{h}{x}+1$. If $x \geq 2$, then it must be that $y \leq h$. If
  $x = 1$, then this inequality gives us that $y \leq h+1$.  However,
  for $y$ to be $h+1$, $W'$ would have to consist of $k-h-1$
  edge-filler and general-filler candidates and $h+1$ dummy
  candidates. Each of these dummy candidates would have to be ranked
  among top $k$ positions by exactly one of the $y$ well-satisfied
  filler voters.
  This would mean that for each edge voter, the only members of $W'$
  ranked by this voter among top $k$ positions would be (some of) the
  edge-filler candidates. Consequently, all the edge voters would rank
  at most $k-h-1$ members of $W'$ among their top $k$ positions. In
  either case (that is, irrespective if $x=1$ or $x \geq 2$), we can upper-bound the score of committee $W'$ by
  assuming that there are $2g(k)\cdot (m+h)\cdot g(k) - h$ voters that assign
  score $g(k-h-1)$ to $W'$ and $2g(k)\cdot m + h$ voters that assign score $g(k)$
  to it.
%
%
  In effect, we have the following inequalities (also see the
  explanations below):
  %
  \begin{align*}
    \score(W') & \leq \big(2g(k)\cdot (m+h) \cdot g(k) - h)\big) \cdot g(k-h-1) + (2g(k)\cdot m+h) \cdot g(k) \\
               & =    2g(k)\cdot (m+h) \cdot g(k) \cdot g(k-h-1) - h \cdot g(k-h-1) + (2g(k)\cdot m+h) \cdot g(k) \\
               & <    2g(k)\cdot (m+h) \cdot g(k) \cdot \big(g(k-h)-1\big) - h \cdot g(k-h-1) + (2g(k)\cdot m+h) \cdot g(k) \\
               & =    T_1 - 2g(k) \cdot (m+h) \cdot g(k) - h \cdot g(k-h-1) + (2g(k)\cdot m+h) \cdot g(k) \\
               & =    T_1 - 2g(k) \cdot (m+h) \cdot g(k) - h \cdot g(k-h-1) + 2g(k)\cdot m \cdot g(k) +h \cdot g(k) \\
               & =    T_1 - 2g(k) \cdot h \cdot g(k) - h \cdot g(k-h-1) +h \cdot g(k) 
                \leq T_1 < T.
  \end{align*}
  %
  %
  The second inequality holds because $g(k-h) > g( k - h - 1) + 1$ (which
  holds due to the fact that $g$ is convex,
  $g(\singularity{g})-g(\singularity{g}-1) > 1$, and $\singularity{g}
  < k - h$). Further inequalities hold due to simple calculations.
  Due to the above reasoning, we can assume that every committee with
  score at least $T$ contains all the $k-h$ filler candidates.

  Consider some committee that contains all the $k-h$ filler
  candidates. We claim that if this committee contains some dummy
  candidates then there is another committee with a higher score.  Why
  is this so? Assume that the committee contains some $z$ dummy
  candidates $(z \leq h)$. If we simply removed these dummy candidates
  (obtaining a smaller committee) then we would lose at most $z\cdot
  g(k)$ points.  Then, we could bring the committee back to its
  intended side by performing the following operations sufficiently
  many times: Either adding to the committee a single vertex candidate
  (already connected by an edge to one from the committee) or adding
  to the committee two vertex candidates connected by an edge. Each of
  these actions increases the score of the committee by at least
  $2g(k) \big( g(\singularity{g}) - g(\singularity{g}-1) \big) >
  2g(k)$ (because for each edge there are $2g(k)$ corresponding edge
  voters).  Thus, would obtain a committee with a score higher than
  the one we have started with. (Note that, technically, there might
  be no sequence of operations that brings our committee back to size
  $k$, but this would only happen if the graph had too few edges to
  contain a clique of size $h$ and we could recognize that this is the
  case in polynomial time.)

  Let $W$ be some winning committee that contains all the $k-h$ filler
  candidates, and some $h$ vertex candidates (by the above paragraph,
  this committee cannot contain any dummy candidates), and let $r$ be
  the number of edges that connect the vertices corresponding to the
  vertex candidates from $W$.  Let us now calculate the score of $W$.
  The filler candidates provide score $T_1$. The situation regarding
  the edge voters requires more care.

  Each edge voter gets score at least $g(\singularity{g}-2)$ due to
  the edge-filler candidates. For each edge for which at least one
  endpoint is in $W$, we get additional $g(\singularity{g}-1) -
  g(\singularity{g}-2)$ points, and for each edge whose both endpoints
  are in $W$, we get yet additional $g(\singularity{g}) -
  g(\singularity{g}-1)$ points. Thus, the edge voters give $W$ the
  following score (see detailed explanations below):
  \begin{align*}
    \underbrace{ \bigg( 2g(k)  \cdot m \cdot  g(\singularity{g}-2) \bigg)}_{=T_2}  + \underbrace{\bigg( 2g(k) \cdot \delta h \cdot \big( g(\singularity{g}-1) -  g(\singularity{g}-2) \big) \bigg)}_{=T_3} \\
    + \underbrace{\bigg( 2g(k) \cdot r \cdot \big( g(\singularity{g}) - g(\singularity{g}-2)  -2( g(\singularity{g}-1) - g(\singularity{g}-2) ) \big) \bigg)}_{\leq T_4}.
  \end{align*}
  The first main term corresponds to the points all the edge voters
  receive, the second is the correction for edge voters that
  correspond to edges that have at least one endpoint in $W$ (note
  that if for some edge both its endpoints belong to $W$, then we add
  $g(\singularity{g}-1) - g(\singularity{g}-2)$ twice, once for each
  endpoint), and the final term corresponds to the correction for
  edges that have two endpoints in $W$. Let us now explain why this
  final correction is appropriate. Consider some edge voter for an
  edge whose both endpoints are in $W$. For this voter, we account
  $g(\singularity{g}-2)$ points that each edge voter gets, we account
  $g(\singularity{g}-1) - g(\singularity{g}-2)$ points for each of the
  endpoints, and $g(\singularity{g}) - g(\singularity{g}-1) -2(
  g(\singularity{g}-1) - g(\singularity{g}-2) )$ points of the final
  correction. Altogether, this sums up to:
  \begin{align*}
    g(\singularity{g}-2) & + 2(g(\singularity{g}-1) -  g(\singularity{g}-2)) + g(\singularity{g}) - g(\singularity{g}-1) \\
                         & -2( g(\singularity{g}-1) - g(\singularity{g}-2) ) = g( \singularity{g} ).
  \end{align*}
  This means that, indeed, we compute the score of edge voters for
  edges whose both endpoints are in $W$ correctly. The same holds for
  all the other edge voters (and follows directly from the above
  analysis).

  Finally, we note that the score $W$ that we obtain from the edge
  voters is maximized when $r$ is maximized. The maximum value that
  $r$ may have is ${h \choose 2}$, which happens if and only if the
  vertex candidates in $W$ correspond to a clique. Then the score that
  the edge voters provide equals $T_2+T_3+T_4$ and the total score of
  the committee is $T$.

  We conclude, that there exists a committee with score at least $T$
  if and only if the input graph contains a size-$h$ clique.
\end{proof}  

Let us now discuss the assumptions of the theorem, where they come
from and why we believe they are natural (or necessary). 

First, the assumption that the counting functions are computable in
polynomial time is standard and clear. Indeed, it would not be
particularly interesting to seek hardness results if already the
counting functions were hard to compute.

Second, we believe that the assumption that the counting functions
$g_{m,k}$ do not depend on $m$ is reasonable. For example, it is quite
intuitive that adding some candidates that all the voters rank last
should not have any effect on the committee selected by a
top-$k$-counting rule.  (The assumption is also very helpful on the
technical level. Our construction uses a number of dummy candidates
that depends on the values of the counting function. If the values of
the counting function depended on the number of candidates, we might
end up with a very problematic, circular dependence.)

Third, the assumption that there is a constant $c$ such that for each
large enough committee size $k$ we have $k - \singularity{g_k} \geq k
/ c$ says that the function ``shows its convex behavior'' early
enough.  As shown in Proposition~\ref{pro:nearly-perfectionist}, some
assumption of this form is necessary (though there is still a gap,
since the bounds from the theorem and from
Proposition~\ref{pro:nearly-perfectionist} do not match perfectly),
and it is the core of the theorem.
%
%
%

Finally, perhaps the least intuitive assumption in this theorem is the
requirement that for a given committee size $k$, the highest value of
the counting function is polynomially bounded in $k$.
The reason for having it 
is that if the highest value were extremely large (say, exponentially
large with respect to $k$) then, for sufficiently few voters (for
example, polynomially many), the rule might degenerate to a
polynomial-time computable rule (for example, it might resemble the
Perfectionist rule for this case).  Exactly to avoid such problems, in
our proof we use a number of voters that depends on $g_k(k)$. Our
reduction would not run in polynomial time if $g_k(k)$ were
superpolynomial.



A result similar to Theorem~\ref{thm:convex-hard}, but for concave
rules, is possible as well (and, in essence, follows from the proofs
of Skowron, Faliszewski, and Lang~\shortcite{sko-fal-lan:c:collective}
and Aziz et
al.~\shortcite{azi-gas-gud-mac-mat-wal:c:approval-multiwinner}).
Thus, in general, top-$k$-counting functions tend to be $\np$-hard to
compute.  What can we do if we need to use them anyway?  There are
several possibilities.  We consider approximability and
fixed-parameter tractability as possible approaches.

\subsection{Approximability}

First, for concave top-$k$-counting rules we can obtain a
constant-factor approximation algorithm (we deduce it from the result
of Skowron, Faliszewski, and Lang~\cite{sko-fal-lan:c:collective},
which---in essence---boils down to optimizing a submodular function 
using the seminal results of Nemhauser et
al.~\cite{nem-wol-fis:j:submodular}).

\newcommand{\textprosubmodular}{Let $\calR_f$ be a top-$k$-counting
  rule defined through a family $f$ of (polynomial-time computable)
  top-$k$-counting functions $f_{m,k}\colon [m]_k\to \N$ with
  corresponding counting functions $g_{m,k}$ that are concave.
Then there is a polynomial-time algorithm
  that, given an election $E$ and a committee size $k$, computes a committee $W$ of size $k$, whose score, under
$\calR_f$, is at least a
  $(1-\frac{1}{e})$ fraction of the score of the winning committee(s) from $\calR_f(E,k)$}


\begin{theorem}\label{pro:submodular}
  \textprosubmodular
\end{theorem}

\begin{proof}
  This follows from the fact that concave top-$k$-counting rules
  correspond to OWA-based rules that use nonincreasing OWA operators.
  For such rules, there is a $(1-\frac{1}{e})$-approximation algorithm
  for computing the score of the winning committees and for computing
  a committee with such
  score~\cite[Theorem~4]{sko-fal-lan:c:collective}.
\end{proof}

Such a general result for convex counting functions seems
impossible. Let us consider a convex counting function $g_{m,k}(x) =
\max(x-1,0)$ that is nearly identical to the linear counting function used by
Bloc. Let us refer to the top-$k$-counting rule defined by $(g_{m,k})_{k \leq m}$ as
NearlyBloc.
%
%
If
we had a polynomial-time constant-factor approximation algorithm for
NearlyBloc,
we would have a constant-factor approximation algorithm for the
\textsc{Densest at most $K$ Subgraph} problem (abbreviated as
\textsc{DamkS}; see below). Taking into account the results of Khuller
and Saha~\shortcite{sahaDensestSubgraphs}, Raghavendra and
Steurer~\shortcite{conf/stoc/RaghavendraS10}, and
Alon~et~al.~\shortcite{techreport/inaaproxDkS}, this seems very
unlikely.

Before we define the \textsc{DamkS} problem, we need to provide some
notation.  In this paper we reserved the symbols $E$ and $V$ to denote
elections and voter collections, respectively. These symbols are also
commonly used to denote the sets of edges and vertices of graphs. To
avoid confusion, given a graph $G$, we refer to its sets of vertices
and edges as $V(G)$ and $E(G)$, respectively.  The \emph{density} of a
graph $G$ is defined as $\delta = \frac{|E(G)|}{|V(G)|}$.

\begin{definition}
  In the \textsc{Densest at most $K$ Subgraph} problem,
  \textsc{DamkS}, we are given a graph $G$ and we ask for a subgraph
  of $G$ of the highest possible density with at most $K$ vertices. 
\end{definition}

%
%


\newcommand{\textthminapprox}{There is no polynomial-time
  constant-factor approximation algorithm for the problem of computing
  the score of a winning committee under NearlyBloc,
  unless such an algorithm exists for the
  \textsc{DamkS} problem.}

\begin{theorem}\label{thm:inapprox}
  \textthminapprox
\end{theorem}

\begin{proof}
  Let $\theta$ be a positive real, $0 < \theta < 1$.  For the
  sake of contradiction, let us assume that there is a polynomial-time
  algorithm $\mathcal{A}$ that, given an election $E$ and committee
  size $k$, outputs a committee $W$ such that, under NearlyBloc the
  score of $W$ is at least an $\theta$ fraction of the score of the
  winning committee.  Using $\mathcal{A}$, we will derive an
  $\frac{\theta}{2}$-approximation algorithm for the \textsc{DamkS}
  problem.

  Let $I$ be an instance of the \textsc{DamkS} problem with a graph $G$
  and an integer $K$. Our algorithm proceeds as follows. For each $B$,
  $1 \leq B \leq K$, we form an election $E_B = (C_B,V_B)$ where:
  \begin{enumerate}

  \item The set of candidates is $C_B = V(G) \cup \bigcup_{e \in E(G)}
    D_e$, where for each $e \in E(G)$, $D_e = \{d_{e, 1}, \ldots d_{e,
      B-2}\}$ is the set of dummy candidates needed for
    our construction.

  \item The collection $V_B$ of voters is such that for each edge $e =
    \{u_1,u_2\} \in E(G)$ we have exactly one voter with preference
    order of the form $\{u_1,u_2\} \pref D_e \pref \cdots$.

  \end{enumerate}
  For each election $E_B$, we run algorithm $\mathcal{A}$ to find a
  committee $W_B$ of size $B$. Each such committee $W_B$ generates an
  induced graph $G_B$ with the vertex set $V(G) \cap W_B$.  We let
  $G_0$ be the trivial subgraph of $G$ consisting of two vertices and
  their connecting edge (if $G$ had no edges, then we could output a
  trivial optimal solution at this point).  We output the densest
  graph among $G_0, G_1, \ldots, G_K$.

  Let us now argue that the above algorithm is an
  $\frac{\theta}{2}$-approximation algorithm for the \textsc{DamkS}
  problem.  Let $\mathrm{OPT}$ be an optimal solution for $I$, with
  the densest subgraph $G'$ consisting of $B$ vertices and $X$
  edges. By definition, $G'$ has density $\delta = \frac{X}{B}$. For each $B$ let
  us consider two cases:
  \begin{description}
  
  \item[Case 1: $\boldsymbol{X \leq \frac{B}{\theta}}$.] In this case,
    the density of the optimal graph is at most equal to
    $\frac{1}{\theta}$. However, a trivial solution with two vertices
    connected with an edge has density equal to $\frac{1}{2}$. Thus,
    in this case this trivial solution is
    $\frac{\theta}{2}$-approximate.

  \item[Case 2: $\boldsymbol{X > \frac{B}{\theta}}$.]  In this case
   we know that there exists a size-$B$ committee for
    election $E_B$ with score at least $X$. Indeed, the committee that
    consists of the vertices from $G'$ obtains one point for
    each edge from $G'$ and has score $X$. Thus
    $\mathcal{A}$ for $E_B$ and committee size $B$ outputs a committee
    $W'$ with score at least $\theta X$. Let $U' = W' \cap V(G)$
    (that is, let $U'$ be the part of this committee that consists of the
    vertex candidates) and let $D' = W' - U'$ (that is, let $D'$ be the
    set of dummy candidates from $W'$).  We observe that the graph
    induced by $U'$ has at least $\theta X - |D'|$ edges. To see this,
    note that since each dummy candidate is ranked among top $B$
    positions by exactly one voter, removing a dummy candidate from
    the committee---in effect decreasing the committee
    size---decreases the total score by at most one. Thus the
    committee consisting only of candidates from $U'$ has score at
    least $\theta X - |D'|$ and each of the points obtained by this
    committee comes from an edge between some members of $U'$.

    The graph induced by $U'$ has density $\delta'$ such that:
    \begin{align*}
      \delta' = \frac{\theta X - |D'|}{B - |D'|} 
               = \frac{\theta X}{B} \cdot \frac{B(\theta X - |D'|)}{\theta X \cdot (B - |D'|)} 
               = \theta \delta \cdot \frac{B(\theta X - |D'|)}{\theta X \cdot (B - |D'|)}
               \geq \theta \delta,
    \end{align*}
    where the last inequality follows from the assumption that $B <
    \theta x$. Indeed, note that:
    \[ 
       B(\theta X - |D'|) = \theta X B - B|D'| \geq \theta X B - \theta X |D'| = \theta X \cdot (B - |D'|).
    \]
    By our assumptions, one
    of these conditions must hold. This means that the graph induced
    by $U'$ is an $\theta$-approximate solution for $I$.
    
  \end{description}
  Since in both cases we obtain at least
  $\frac{\theta}{2}$-approximate solutions, our algorithm is
  $\frac{\theta}{2}$-approximate. Since it is clear that it runs in
  polynomial time, the proof is complete.
\end{proof}

Nonetheless, for top-$k$-counting rules that are not too far from
$\alpha_k$-CC, we have a polynomial-time approximation scheme (PTAS),
that is, an algorithm that can achieve any desired approximation
ratio, as long as the number of candidates is not too large relative
to the committee size. This result holds even for rules that are not
concave (provided they satisfy the conditions of the theorem); the
result follows by noting that our voters have non-finicky
utilities~\cite{sko-fal-lan:c:collective}.

\newcommand{\textthmptas}{Let $\calR_f$ be a top-$k$-counting committee scoring rule, where the family $f=(f_{m,k})_{k\le m}$ is
  defined through a family of counting functions $(g_{m,k})_{k \leq m}$ that are: (a)
  polynomial-time computable and (b) constant for arguments
  greater than some given value $\ell$. If $m
  = o(k^2)$, there is a PTAS for computing
  the score of a winning committee under~$\calR_f$.}

\begin{theorem}\label{thm:ptas}
  \textthmptas
\end{theorem}

\begin{proof}
  We use the concept of non-finicky utilities provided by Skowron et
  al.~\shortcite{sko-fal-lan:c:collective}. Adapting their
  terminology, we say that a single-winner scoring function $\gamma_m\colon [m]\to \N$
  (for elections with $m$ candidates) is $(\xi,\delta)$-non-finicky for $\xi, \delta
  \in [0,1]$, if each of the highest $\lceil \delta m \rceil$
  numbers in the sequence $\gamma_m(1), \ldots, \gamma_m(m)$ is greater or
  equal to $\xi \gamma_m(1)$. It is easy to see that $\alpha_k$ is
  $(1,\frac{k}{m})$-non-finicky.

  Consider an input election $E = (C,V)$ with $m$ candidates, and
  committee size $k$, such that $m = o(k^2)$. 
  By Proposition~\ref{pro:owa}, we know that $f_{m,k}$ is OWA-based,
  that it uses some OWA operator $\Lambda_{m, k}$ that has nonzero
  entries on the top $\ell$ positions only, and that it uses scoring
  function $\alpha_k$ (which is a $(1,
  \frac{k}{m})$-non-finicky). Thus, due to Skowron et
  al.~\shortcite{sko-fal-lan:c:collective}, there is a polynomial-time
  $\left(1 - \ell
    \exp\left(-\frac{k^2}{m\ell^2}\right)\right)$-approximation
  algorithm for computing the score of a winning committee under
  $f$.\footnote{Strictly speaking, the exact formulation of the result
    that we invoke here appears only in the full version of their
    work~\cite[Theorem~29]{sko-fal-lan:c:collective-fulltext}, and not
    in the conference extended abstract.} Using the assumption that $m
  = o(k^2)$, the approximation ratio of the algorithm is:
  \begin{align*}
    \alpha & = 1 - \ell \exp\left(-\frac{k^2}{m\ell^2}\right) = 1 - \ell \exp\left(-\frac{k^2}{o(k^2)\ell^2}\right) 
    = 1 - \ell \exp\left(-\frac{1}{o(1)}\right) = 1 - o(1).
  \end{align*}
  This completes the proof.
\end{proof}

Theorem~\ref{thm:ptas} is quite remarkable even for the case of
$\alpha_k$-CC (let alone that it applies to a somewhat more general
set of rules). Indeed, generally, variants of the Chamberlin--Courant
rule that use some sort of approval scoring function are hard to
compute~\cite{pro-ros-zoh:j:proportional-representation,bet-sli-uhl:j:mon-cc}
and the best possible approximation ratio for a polynomial-time
algorithm, in the general case, is $1-\frac{1}{e}$ (this result was
observed by Skowron and Faliszewski~\cite{sko-fal:c:maxcover} and
follows from results for the MaxCover problem~\cite{fei:j:cover}).
However, this upper bound relies on the fact that there is no
connection between the size of the input election, the committee size,
and the number of candidates that each voter approves. We obtain a
PTAS because we assume that for the committee size $k$ each voter
approves of $k$ candidates, and that the number $m$ of candidates is
such that $m = o(k^2)$.

One may ask how likely it is that this last assumption holds.
As a piece of anecdotal evidence, we mention that in the 2015
parliamentary elections in Poland, there were $k=460$ seats in the
parliament and $m \approx 8000$ candidates. In this case, $m/{k^2}
\approx 0.0378$, which suggests that our algorithm could be effective
(provided that the voters could say which $k$ candidates they approve
of; likely, this would require some sort of simplified ballots, for
example, allowing one to approve blocks of candidates).

\subsection{Fixed-Parameter Tractability}

If one were not interested in approximation algorithms but still
wanted to use top-$k$-counting rules, then one might seek
fixed-parameter tractable algorithms.  In parameterized complexity we
concentrate on some distinguished parameter in the problem instances,
such as the number of candidates or the number of voters. We say that
a parameterized problem is fixed-parameter tractable (is in $\fpt$) if
there is an algorithm that, given an instance of this problem of size
$n$ with parameter $t$, computes an answer for the problem in time
$f(t)n^{O(1)}$, where $f$ is some computable function (such an
algorithm is also said to run in $\fpt$-time with respect to parameter
$t$).  For a detailed description of parameterized complexity, we
point the readers to the books of Downey and
Fellows~\shortcite{dow-fel:b:parameterized},
Niedermeier~\shortcite{nie:b:invitation-fpt}, and Cygan et
al.~\cite{cyg-fom-kow-lok-mar-pil-pil-sau:b:fpt}.

We start with a very simple observation, namely that a winning committee can
be computed for every top-$k$-counting rule in $\fpt$ time for the
parameterization by the number of candidates.

\begin{proposition}\label{prop:fptCandidates}
  Let $\calR_f$ be a top-$k$-counting committee scoring rule, where
  the family $f=(f_{m,k})_{k\le m}$ is defined through a family of
  counting functions $(g_{m,k})_{k \leq m}$ (that are computable in
  $\fpt$ time with respect to $m$).  There is an algorithm that, given
  a committee size $k$ and an election $E$, computes a winning
  committee from $\calR_f(E,k)$ in $\fpt$-time with respect to the
  number $m$ of candidates.

%
\end{proposition}

\begin{proof}
  The algorithm simply computes the score of every possible committee
  and outputs the one with the highest score. With $m$ candidates and
  committee size $k$, the algorithm has to check $\binom{m}{k} =
  O(m^m)$ committees, and checking each committee requires $\fpt$ time
  only.
\end{proof}

For rules based on concave counting functions we can also provide a
far less trivial $\fpt$ algorithm for the parameterization by the
number of voters.

\newcommand{\textthmfptn}{
Let $\calR_f$ be a top-$k$-counting committee scoring rule, where the family $f=(f_{m,k})_{k\le m}$ is
  defined through a family of concave counting functions $(g_{m,k})_{k \leq m}$.  There is an algorithm that, given a committee size $k$ 
and an election $E$, computes a winning committee from $\calR_f(E,k)$
 in $\fpt$-time with respect to the number $n$ of voters.}

\begin{theorem}\label{thm:fpt-n}
  \textthmfptn
\end{theorem}

\newcommand{\tttype}{\mathcal{T}}

\begin{proof}
  Our algorithm is based on solving a mixed integer linear program
  (MILP) in $\fpt$-time with respect to the number of integral
  variables. The key trick is to use non-integral variables in such a
  way that in every optimal solution they have to take integral values
  (this technique was first used by Bredereck et
  al.~\shortcite{bre-fal-nie-sko-tal:c:covering-milp}).
  
  Let $k$ be the input committee size and $E = (C,V)$ be
  the input election, where $C = \{c_1, \ldots, c_m\}$ is the set of
  candidates, $V = (v_1, \dots, v_n)$ is the collection of voters.

  We enumerate all the nonempty subsets of $V$ as $S_1, \ldots,
  S_{2^n-1}$.  For each $i \in [2^n-1]$, let $\tttype(S_i)$ denote the
  largest set of candidates that satisfies the following condition:
  Every voter in $S_i$ ranks each candidate from $\tttype(S_i)$ among
  the top $k$ positions and no other voter ranks either of the
  candidates from $\tttype(S_i)$ among top $k$ positions. Note that
  $\tttype(S_1), \ldots, \tttype(S_{2^n})$ is a partition of $C$.
  We illustrate this partition in the following example.
  \begin{example}
    Consider an election $E = (C,V)$ with $C = \{a,b,c,d,e,f\}$ and $V
    = (v_1, \ldots, v_6)$, where the voters have the following
    preference orders (we set the committee size $k = 3$ and, thus, we
    list only top $k$ positions for each vote):
    \newcommand{\tvote}[3]{#1 \succ #2 \succ #3 \succ \cdots}
    \begin{align*}
      v_1 \colon & \tvote c d f, &
      v_2 \colon & \tvote c d e, &
      v_3 \colon & \tvote a b c, \\
      v_4 \colon & \tvote c e f, &
      v_5 \colon & \tvote d e f, &
      v_6 \colon & \tvote a b e.
    \end{align*}
    We have the following sets: $\tttype(\{v_3,v_6\}) = \{a,b\}$
    since only voters $v_3$ and $v_6$ rank $a$ and $b$ on top three
    positions (and there are no other candidates they both rank among
    their top three positions). Then, we have:
    $\tttype(\{v_1,v_2,v_3,v_4\}) = \{c\}$, $\tttype(\{v_1,v_2,v_5\})
    = \{d\}$, $\tttype(\{v_2,v_4,v_5, v_6\}) = \{e\}$, and
    $\tttype(\{v_1,v_4,v_5\}) = \{f\}$. For every other subset $S_i$
    of voters, we have $\tttype(S_i) = \emptyset$.  For example,
    $\tttype\{v_4,v_5\} = \emptyset$ for the following reasons: The
    candidates that both $v_4$ and $v_5$ rank on top three positions
    are $e$ and $f$. However, each of these candidates is ranked among
    top three positions also by some other voter(s).
  \end{example}


  \begin{figure}[t!]

    {\upshape
      \begin{align*}
        \intertext{\quad\quad\quad\quad maximize $\sum_{i=1}^n
          \sum_{j=1}^k x_{i, j} \cdot (g_{m,k}(j) - g_{m,k}(j - 1))$}
        \intertext{\quad\quad\quad\quad subject to: }
        &\text{(a)}\ \ \sum_{i=1}^{2^n-1} z_i = k,                                          \ &\\
        &\text{(b)}\ \ x_{i} = \sum_{j \colon i \in S_j} z_j,                                     \ & i \in [n] \\
        &\text{(c)}\ \ \sum_{j=1}^k x_{i, j} = x_i,                                      \ & i \in [n] \\
        &\text{(d)}\ \ 0 \leq x_{i, j} \leq 1,                                           \ &  i \in [n]; j \in [k] \\
        &\text{(e)}\ \ 0 \leq z_{i} \leq |\tttype(S_i)|, \ & i \in
        [2^n-1]
      \end{align*}
    }
    \caption{\label{fig:milp}The Mixed Integer Linear Program used in the
      proof of Theorem~\ref{thm:fpt-n}.}
      
  \end{figure}

  Our algorithm forms a mixed integer linear program with the
  following variables.  We have $2^n-1$ integer variables, $z_1,
  \ldots z_{2^n-1}$,  where, intuitively, each $z_i$ describes how
  many candidates from the set $\tttype(S_i)$ we take into the winning
  committee.  For each $i \in [n]$ we also have an integer variable
  $x_i$, which describes how many candidates from the top $k$
  positions of the preference order of voter $v_i$ belongs to the
  winning committee. Finally, for each variable $x_i$, we have
  rational variables $x_{i,j}$, $0 \leq x_{i,j} \leq 1$, such that
  (intuitively) each $x_{i,j}$ is $1$ if $x_i$ is at least $j$.  We
  present our mixed integer linear program in Figure~\ref{fig:milp}.
  To solve this program, we invoke Lenstra's famous result in its
  variant for mixed integer
  programming~\cite[Section~5]{len:j:integer-fixed}.

  Now it remains to argue that it indeed outputs a correct solution,
  that is, that the variables $z_1, \ldots, z_{2^n-1}$ describe a
  winning committee. If all the variables have the intended, intuitive
  values (as described in the preceding paragraph), then---with our
  maximization goal in mind---one can verify that variables $z_1,
  \ldots, z_{2^n-1}$ describe a winning committee. Thus we show that,
  indeed, all the variables have their intended values.

  Due to constraints (a) and (e), variables $z_1, \ldots, z_{2^n-1}$
  certainly describe a possible committee of size $k$ (from each set
  $\tttype(S_i)$ we take $z_i$ arbitrary candidates). Constraints (b)
  ensure the correct values of variables $x_1, \ldots,
  x_{n}$. Finally, the maximization goal and constraints (c) ensure
  that each variable $x_{i,j}$ is $1$ exactly if $x_i \geq j$ and is
  $0$ otherwise. This is so, because $g_{m,k}$ is concave. Thus, if for some
  values $j$ and $j'$ with $j < j'$ it was the case that $x_{i,j} < 1$
  and $x_{i,j'} > 0$ then increasing $x_{i,j}$ and decreasing
  $x_{i,j'}$ by the same amount (without breaking constraint (d))
  would yield a higher value of the function to be maximized.
\end{proof}

To summarize, it appears that most (but certainly not all)
top-$k$-counting rules are $\np$-hard to compute. For top-$k$-counting
rules based on concave counting functions, there are good
polynomial-time approximation algorithms and some exact $\fpt$
algorithms. On the other hand, for rules based on convex functions the
situation is much more difficult. Aside from several algorithms that
do not depend on concavity or convexity of the counting problem (for instance the algorithms from Theorem~\ref{thm:ptas} and Proposition~\ref{prop:fptCandidates}), so
far we only have evidence for hardness of approximation.

\section{Conclusions and Further Research}

Aiming at finding a multiwinner analogue of the single-winner
Plurality rule, we have shown that the answer is quite
involved. While intuitively SNTV is a natural analogue of 
Plurality, it fails the fixed-majority criterion (which
Plurality satisfies in the single-winner setting). We have found
that, among all committee scoring rules, only the top-$k$-counting
rules---a class of rules we have defined in this paper---have a chance
of satisfying our criterion, and we have characterized exactly when
this happens. Specifically, we have shown that the committee scoring rules which satisfy the fixed-majority
criterion are exactly those top-$k$-counting rules whose counting functions satisfy a relaxed
variant of convexity.

%
%
For example, the Bloc and Perfectionist rules both satisfy the
fixed-majority criterion and, so, in some sense, they are among the
multiwinner analogues of Plurality (for the Perfectionist
rule this goes quite deep). 
 On the other hand, a variant of the Chamberlin--Courant rule based on
$k$-Approval scoring function is top-$k$-counting, but fails the fixed-majority
criterion.

We believe that it is most interesting to focus on top-$k$-counting
rules based either on convex or on concave counting functions. These
two classes of rules are different in some interesting way. On the
one hand, top-$k$-counting rules based on convex counting functions
are fixed-majority consistent, but seem very hard to compute (with a
few exceptions). On the other hand, top-$k$-counting rules based on
concave counting functions fail the fixed-majority criterion (the
borderline case of Bloc rule excluded), but are much easier to
compute (typically still $\np$-hard, but with constant-factor
polynomial-time approximation algorithms and $\fpt$ algorithms for the
parameterization by the number of voters).


Our work leads to a number of open questions. In the axiomatic
direction, it would be interesting to provide a characterization of
committee scoring rules along the lines of Young's characterization
for their single-winner
counterparts~\shortcite{you:j:scoring-functions}. On the computational
front, it would be interesting to find more powerful algorithms for
computing winning committees under various top-$k$-counting rules
(e.g., for the $\alpha_k$-PAV rule).

\section{Acknowledgments}  

Piotr Faliszewski and Piotr Skowron were supported in part by NCN grant
DEC-2012/06/M/ST1/00358. Piotr Faliszewski was supported in part by
AGH University grant 11.11.230.124 (statutory research) and Piotr Skowron was partially supported by ERC-StG639945.
Arkadii Slinko was  supported by the Royal Society of NZ  
Marsden Fund UOA-254. 
Nimrod Talmon was supported by the DFG Research Training Group MDS (GRK 1408).
The research was also partially supported through the COST action IC1205 (Piotr
Faliszewski's visit to TU Berlin).

\bibliography{grypiotr2006}

\end{document}